\documentclass[journal,onecolumn]{IEEEtran}
\usepackage{cite}

\usepackage{perlaza}
\usepackage{tikz}
\usepackage{amsmath}
\usepackage[hidelinks]{hyperref}
\usepackage{url}
\usepackage{pifont}
\usepackage{longtable}
\usepackage{tabularx}
\usepackage{makecell, cellspace, caption}
\usepackage{array}
\newcolumntype{L}[1]{>{\raggedright\let\newline\\\arraybackslash\hspace{0pt}}m{#1}}
\newcolumntype{C}[1]{>{\centering\let\newline\\\arraybackslash\hspace{0pt}}m{#1}}
\newcolumntype{R}[1]{>{\raggedleft\let\newline\\\arraybackslash\hspace{0pt}}m{#1}}

\newcommand{\todo}[2][]{\ignorespaces
	\if\relax\detokenize{#1}\relax
	{\color{red}[TODO: #2]}%
	\else
	{\color{red}[TODO (#1): #2]}%
	\fi
}

\definecolor{change}{HTML}{0096b8}

\makeatletter
\DeclareFontFamily{U}{tipa}{}
\DeclareFontShape{U}{tipa}{m}{n}{<->tipa10}{}
\newcommand{\arc@char}{{\usefont{U}{tipa}{m}{n}\symbol{62}}}%

\newcommand{\arc}[1]{\mathpalette\arc@arc{#1}}

\newcommand{\arc@arc}[2]{%
  \sbox0{$\m@th#1#2$}%
  \vbox{
    \hbox{\resizebox{\wd0}{\height}{\arc@char}}
    \nointerlineskip
    \box0
  }%
}
\makeatother
\newcommand{\ts}{\textsuperscript}

\DeclareMathOperator{\sinc}{sinc}

\begin{document}
\title{Simultaneous Information and Energy Transmission with Short Packets and Finite Constellations}

\author{Sadaf ul Zuhra,~\IEEEmembership{Member,~IEEE,}
Samir M. Perlaza,~\IEEEmembership{Senior Member,~IEEE,}
H. Vincent Poor,~\IEEEmembership{Fellow,~IEEE,}
Mikael Skoglund,~\IEEEmembership{Fellow,~IEEE}
\thanks{Sadaf ul Zuhra, Samir M. Perlaza, and H. Vincent Poor are with the Department of Electrical and Computer Engineering, Princeton University, 08540 Princeton, NJ, USA. $\lbrace$sadaf.zuhra, poor$\rbrace$@princeton.edu\newline
Samir M. Perlaza is with INRIA, Centre Inria d'Universit\'{e} C\^{o}te d'Azur, 2004  Route des Lucioles, 06902 Sophia Antipolis, France. $\lbrace$samir.perlaza$\rbrace$@inria.fr\newline
Mikael Skoglund is with the School of Electrical Engineering and Computer Science, Malvinas V\"{a}g 10,  KTH Royal Institute of Technology, 11428 Stockholm, Sweden. (skoglund@kth.se)\newline
Samir M. Perlaza is also with the Laboratoire de Math\'{e}matiques GAATI, Universit\'{e} de la Polyn\'{e}sie Fran\c{c}aise,  BP 6570, 98702 Faaa, French Polynesia. \newline
%
%
This research was supported in part by the European Commission through the H2020-MSCA-RISE-2019 program under grant 872172; in part by the Agence Nationale de la Recherche (ANR) through the project MAESTRO-5G (ANR-18-CE25-0012); in part by the U.S. National Science Foundation under Grant CCF-1908308; and in part by the French Government through the ``Plan de Relance" and ``Programme d’investissements d’avenir". \newline
This paper was presented in part at the IEEE Information Theory Workshop (ITW) 2021, IEEE International Symposium on Information Theory (ISIT)~2022, and IEEE ITW 2022.
}
}

%

\maketitle

\begin{abstract}
This paper characterizes the trade-offs between information and energy transmission over an additive white Gaussian noise channel in the finite block-length regime with finite channel input symbols.
These trade-offs are characterized in the form of inequalities involving the information transmission rate, energy transmission rate, decoding error probability (DEP) and energy outage probability (EOP) for a given finite block-length code.
The first set of results identify the tuples of information rate, energy rate, DEP and EOP that cannot be achieved by a given code.
Following this, a novel method for constructing a family of codes that can satisfy a target information rate, energy rate, DEP and EOP is proposed.
Finally, the achievability results identify the set of tuples of information rate, energy rate, DEP and EOP that can be simultaneously achieved by the constructed family of codes.
\end{abstract}


\section{Introduction} \label{sec:intro}
Simultaneous information and energy transmission (SIET), also known as simultaneous wireless information and power transfer (SWIPT), is a technology that employs radio frequency (RF) signals to simultaneously transmit both information and energy to (possibly different) remote devices~\cite{Tesla-Patent-1914, GroverSahai,Iqbal2024}.
A central question in modern SIET is to characterize the information rate, energy rate, decoding error probability (DEP) and energy outage probability (EOP) that can be simultaneously achieved by any given transmission technique.  
Such question has been traditionally studied under at least one of two simplifying assumptions. 
The first assumption is that the transmission duration is infinitely long, as in~\cite{varshney2008transporting, amor2016fundamental, NizarItw2021, GroverSahai, amor2016feedback, KhalfetGIC}. Under this assumption, the DEP and the EOP can be made arbitrarily close to zero, and thus, the focus is exclusively on the information and energy transmission rates. If this assumption is dropped, as in the remainder of this work, the DEP and EOP are bounded away from zero, and thus, the performance analysis of SIET shall take into account four metrics: information rate, energy rate, DEP, and EOP. 
%
The second common assumption is that channel inputs symbols can be drawn from probability distributions whose support is uncountable, e.g., Gaussian distributions as in~\cite{8115220}. This presumes the ability to transmit with infinitely many channel input symbols, which is not the case in practice. In a departure from this assumption, this work considers the more practical case of finite channel input symbols, which provides insights for practical systems operating with modulation schemes built upon commercial constellations.

\subsection{State of the art}
The existing body of work in SIET can be divided into three main categories. The first category is the study of the trade-offs between the information and energy rates that can be simultaneously transmitted by an RF signal. In the asymptotic block-length regime, the focus is on the notion of information-energy regions to characterize the set of all information and energy rate pairs that can be simultaneously achieved~\cite{amor2016fundamental}. To capture the trade-off between the information and energy rates, a capacity-energy function is defined in~\cite{varshney2008transporting} for various channels including the discrete memoryless channel, binary symmetric channel, and the additive white Gaussian noise (AWGN) channel. The information-energy capacity region of the Gaussian multiple access channel is characterized in~\cite{amor2016feedback,amor2017feedback}, whereas the information-energy capacity region of the Gaussian interference channel is approximated in~\cite{KhalfetGIC}. In the finite block-length regime,~\cite{perlaza2018simultaneous} and~\cite{khalfet2019ultra} provide a characterization of the information-energy capacity region with binary antipodal channel inputs and fixed DEP and EOP. Converse and achievability bounds on these metrics for finite channel input symbols are presented in~\cite{zuhraITW} and~\cite{zuhraISIT}, respectively. The fundamental limits of SIET for the case of semantic communications are characterized in~\cite{10613795}.
%


The second category of research studies the modeling of energy harvesting (EH) circuits. One of the key questions in this regard tackles the estimation of the amount of  energy that can be harvested from an RF signal. This line of inquiry has revealed that, due to the presence of non-linear elements such as diodes in the EH circuits, the expected energy harvested from a signal is a function of the second and fourth powers of the signal magnitude~\cite{varasteh2017wireless}. Recent research on EH non-linearities~\cite{8115220,7547357} has shown that energy models that do not account for these non-linearities result in inaccurate estimates of the harvested energy. The impact of the energy harvester non-linearities on the fundamental limits of SIET in the finite block-length regime has been studied in~\cite{zuhraITW2}. 
In~\cite{9241856}, a learning based approach is proposed for modeling the EH circuits where the memory of non-linear elements in the EH circuit is modeled as a Markov decision process.

The third line of research covers the system design, implementation, and applications of SIET networks.
Optimal waveform design for SIET from a multi-antenna transmitter to multiple single antenna receivers is studied in~\cite{9447959}.
Signal and system design exclusively for wireless energy transmission has been studied in~\cite{7547357,9411899,9184149,7867826, 9447237} and~\cite{9153166}. In~\cite{9377479}, the authors optimize resource allocation and beamforming for intelligent reflecting surfaces aided SIET. 
 An algorithm for designing a circular quadrature amplitude modulation scheme for SIET that maximizes the peak-to-average power ratio has been proposed in~\cite{9593249}. 
 
More comprehensive overviews of the work on SIET in the second and third categories detailed above, can be found in~\cite{survey,surveyA},~\cite{8476597}, and~\cite{ClerckxFoundations}. A comparison of relevant aspects of the existing literature on SIET with this work is provided in Table~\ref{TableSOA} below. A tick mark indicates that the specific factor has been taken into consideration while a dash indicates the opposite.

\begin{table*}[t] 
\captionof{table}{Summary of the state-of-art} 
\centering
\begin{tabular}{ | m{1.5 cm} | m{1.3 cm} |m{1.7 cm} |m{1.5 cm} | m{3cm}| m{1cm}| m{1cm}|} 
 \hline
  \textbf{Reference}  & \textbf{Channel inputs} & \textbf{Block-length} & \textbf{EH non-linearities} & \textbf{Channel} & \textbf{DEP} & \textbf{EOP} \\
 \hline
 \hline
\cite{varshney2008transporting}       & Infinite                    & Asymptotic			& -				& DMC, AWGN & - & -\\
\hline
\cite{GroverSahai}  		      &Infinite			 & 	Asymptotic			& 	-			&  AWGN + Frequency-selective fading & - & - \\
\hline
\cite{NizarItw2021} 			& Infinite                    &     Asymptotic                            & \checkmark 		& Rayleigh fading & - & - \\
\hline
\cite{6373669} 				& Infinite 	               & Asymptotic				& 	-			& Flat fading & - & - \\
\hline
\cite{9241856} 		                & Infinite		       	       & Asymptotic		& -				& AWGN& - & - \\
 \hline
 \cite{9149424} 		        & Finite		       	       & 	Asymptotic			& -				& AWGN + Fading& - & - \\
 \hline
 \cite{9377479} 		        & Finite		       	       & 	Asymptotic			& -				& AWGN + Flat-fading& - & - \\
 \hline
 \cite{7998252} 		        & Finite		       	              & -			& 	-			& AWGN + Rayleigh fading & \checkmark & - \\
 \hline
  \cite{9593249} 		& Finite		       	              & -			& -				& AWGN& -  & - \\
 \hline
    \cite{varasteh2017wireless} 	& Infinite		       	              & Asymptotic	& \checkmark			& AWGN& - & - \\
 \hline
     \cite{8115220} 		        & Infinite		       	              & Asymptotic	& \checkmark			& Multi-path fading & - & - \\
     \hline
   \cite{9741251}			&	Uncountable	&	Asymptotic	&	-		&	AWGN + Rayleigh fading	& \checkmark & -\\
  \hline
   \cite{9734045}			&	Uncountable	&	Asymptotic	& -		&	AWGN + Rayleigh and Rician fading	& - & - \\
  \hline
   \cite{liu2022joint}			&	Uncountable	&	Asymptotic	&	-		&	AWGN	& -& -\\
   \hline
   \cite{9502719}			&	Uncountable	&	Asymptotic	&	\checkmark		&	Multi-path fading	& - & -\\     
 \hline
    \cite{6489506}			&	Uncountable	&	Asymptotic	&	-		&	Multi-path fading	& - & -\\     
 \hline
     \cite{ 7063588}			&	Uncountable	&	Asymptotic	&	-		&	AWGN + Flat-fading	& - & -\\     
 \hline
      \cite{9169700 }			&	Finite	&	Finite	&	 \checkmark			&	AWGN 	& - & -\\     
 \hline
  \cite{perlaza2018simultaneous}               & Finite                         & Finite              & 	-			& BSC & \checkmark & \checkmark\\
\hline
 \cite{khalfet2019ultra} 	& Finite		       	              & Finite			& 	-			& BSC& - & -\\
 \hline
      \textbf{This work} 	& Finite		       	              & Finite	& \checkmark			& AWGN& \checkmark & \checkmark\\
        \hline
 \end{tabular} \label{TableSOA}
 \end{table*}

\subsection{Contributions}
This work provides a comprehensive characterization of the relationships between the parameters of short packet codes with finite channel input symbols (finite constellations, e.g., QPSK, 16-QAM, 64-QAM, among others) for simultaneous transmission of information and energy.
The main contributions can be summarized as follows.
\begin{itemize}
\item The first set of results, presented in Section~\ref{sec:results}, provide bounds on the information transmission rate, energy transmission rate, DEP and EOP of any given code over a point-to-point AWGN channel.
All these bounds emerge to be functions of common code parameters such as the block-length and the type.
Together these bounds form the necessary conditions that any code must satisfy for accomplishing the tasks of information and energy transmission simultaneously.
Through the dependence on common parameters of the code, these bounds also characterize the trade-offs between information rate, energy rate, DEP and EOP of a given code.
\item A method for constructing codes for SIET over short packet communication is proposed.
Using this method, a family of codes can be designed to achieve given rates of information and energy transmission while also ensuring that the DEP and EOP are kept within the required thresholds.
The feasible tuples of information rate, energy rate, DEP and EOP that can be achieved by the codes thus constructed are also characterized.
\item The performance of the constructed codes is evaluated by comparing the achievability results with the identified necessary conditions.
Comprehensive numerical analyses are also conducted to illustrate the trade-offs between the competing objectives of information and energy transmission identified by the results in this work.
\end{itemize}

%
  
\subsection{Notation} \label{SecNotation}
The sets of natural, real and complex numbers are denoted by $\ints$, $\reals$ and $\complex$, respectively. In particular, $0 \notin \ints$.
Random variables and random vectors are denoted by uppercase letters and uppercase bold letters, respectively. Scalars are denoted by lowercase letters and vectors by lowercase bold letters. The real and imaginary parts of a complex number $c \in \complex$ are denoted by  $\Re(c)$ and $\Im(c)$, respectively. The complex conjugate of $c \in \complex$ is denoted by $c^\star$ and the magnitude of $c$ is denoted by $|c|$.  The imaginary unit is denoted by $\mathrm{i}$, \ie, $\mathrm{i}^2 = -1$.  
The empty set is denoted by $\emptyset$. The $\sinc$ function is defined as follows
\begin{IEEEeqnarray}{rCl} \label{EqSinc}
\sinc(t) \triangleq \frac{\sin(\pi t)}{\pi t}.
\end{IEEEeqnarray}
The tail distribution function of the standard normal distribution is referred to as the $Q$-function where 
\begin{IEEEeqnarray}{rCl} \label{DefQfunc}
\mathrm{Q}(x) = \int_x^\infty \frac{1}{\sqrt{2 \pi}} \exp \left(-\frac{t^2}{2} \right) \mathrm{d}t.
\end{IEEEeqnarray}
\section{System Model}  \label{sec:system_model}
Consider a communication system formed by a transmitter, an information receiver (IR), and an energy harvester (EH). The objective of the transmitter is to simultaneously send information to the IR at a rate of $R$ bits per second while also delivering $B$ Joules of energy to the EH over an AWGN channel.
 The transmission takes place over a finite duration of $n \in \ints$ channel uses.
 Thus, the rate of energy transmission is $B/n$ Joules per channel use.
 The transmitter uses $L$ symbols from the set of channel input symbols
\begin{equation}\label{EqCIsymbols}
\mathcal{X} \triangleq \{x^{(1)}, x^{(2)}, \ldots, x^{(L)}\} \subset \mathds{C}.
\end{equation} 
For all $m \inCountK{n}$, denote by $\nu_m \in \mathcal{X}$, the symbol to be transmitted during channel use $m$.
Denote the vector of channel input symbols transmitted over $n$ channel uses by
\begin{IEEEeqnarray}{rCl} \label{Eqnu}
\boldsymbol{\nu} = (\nu_1, \nu_2, \ldots, \nu_n )^{\sf{T}}.
\end{IEEEeqnarray}
The baseband frequency of the transmitter in Hertz (Hz) is $f_w$. Denote by $T \triangleq \frac{1}{f_w}$, the duration of a channel use in time units. Hence, the transmission takes place over $nT$ time units.
The complex baseband signal at time $t$, with $t \in [0,nT]$ is given by
\begin{IEEEeqnarray}{rCl}
x(t) &=& \sum_{m=1}^n  \nu_m \sinc \left( f_w \left(t - (m-1)T\right) \right), \label{Eq4} 
\end{IEEEeqnarray}
where the $\sinc$ function is defined in~\eqref{EqSinc}. The signal $x(t) $ in~\eqref{Eq4} has a bandwidth of $\frac{f_w}{2} > 0$ Hz. Let $f_c > \frac{f_w}{2}$ denote the center frequency of the transmitter. The RF signal input to the channel at time $t$, denoted by $\tilde{x}(t)$, is obtained by the frequency up-conversion of the baseband signal $x(t)$ in~\eqref{Eq4} as follows:
\begin{IEEEeqnarray}{rcl}
\tilde{x}(t) &=& \Re \left(x(t) \sqrt{2} \exp(\mathrm{i} 2 \pi f_c t) \right), \label{EqXrf} 
\end{IEEEeqnarray}
where $\mathrm{i}$ is the complex unit.
The RF outputs of the AWGN channel at time $t \in [0,nT]$ are the random variables
\begin{subequations}
\begin{IEEEeqnarray}{rCl}
   \label{EqYt} Y(t) & = & \tilde{x}(t) + N_1(t), \mbox{ and }  \\
    Z(t) & = & \tilde{x}(t) + N_2(t), \label{Eq7b}
\end{IEEEeqnarray}
\end{subequations}
where, for all $t \in [0,nT]$, the random variables $N_1(t)$ and $N_2(t)$ represent real white Gaussian noise with zero mean and variance $\sigma^2$; and the random variables $Y(t)$ and $Z(t)$ are the inputs to the IR and the EH, respectively. 

At the IR, the received signal $Y(t)$ in~\eqref{EqYt} is first multiplied with $\sqrt{2} \exp(-\mathrm{i} 2 \pi f_c t)$ to obtain the down-converted output. The down-converted output is then passed through a unit gain low pass filter with impulse response $f_w \sinc \left( f_w t \right)$ that has a cut-off frequency of $\frac{f_w}{2}$ Hz to obtain the complex baseband equivalent of $Y(t)$. This is followed by ideally sampling the complex baseband output at intervals of $1/f_w$. 
The resulting discrete time baseband output at the end of $n$ channel uses is given by the following random vector~\cite[Section $2.2.4$]{tseV}:
\begin{IEEEeqnarray}{rCl} \label{EqChannelModel}
{\boldsymbol Y} & = & \boldsymbol{\nu} + {\boldsymbol N},
\end{IEEEeqnarray}
where the vector ${\boldsymbol Y} = (Y_1,Y_2, \ldots, Y_n)^{\sf{T}} \in \mathds{C}^{n}$ is the input to the IR;  the vector of channel input symbols $\boldsymbol{\nu}$ is in~\eqref{Eqnu}; and the noise vector $\boldsymbol{N} = (N_{1}, N_{2}, \ldots, N_{n})^{\sf{T}}\in \mathds{C}^{n}$ is such that, for all $m \inCountK{n}$, the random variable $N_{m}$ is a complex circularly symmetric Gaussian random variable whose real and imaginary parts have zero means and variances $\frac{1}{2}\sigma^2$. Moreover, the random variables $N_1, N_2, \ldots, N_n$ are mutually independent (see~\cite[Section $2.2.4$]{tseV}). 
Therefore, for all $\boldsymbol{y} = (y_1, y_2, \ldots, y_n)^{\sf{T}} \in \mathds{C}^{n}$, and for all $\boldsymbol{\nu} = (\nu_1, \nu_2, \ldots, \nu_n)^{\sf{T}} \in \mathds{C}^{n}$, the conditional probability density function of the channel output $\boldsymbol{Y}$ in~\eqref{EqChannelModel} is given by
\begin{subequations}\label{EqYXdistribution}
\begin{IEEEeqnarray}{rCl} 
    f_{\boldsymbol{Y}|\boldsymbol{X}}(\boldsymbol{y}|\boldsymbol{\nu}) & = & \prod_{m=1}^n f_{Y|X}(y_m|\nu_m),
\end{IEEEeqnarray}
where, for all $m \in \lbrace 1,2, \ldots, n \rbrace$,
\begin{IEEEeqnarray}{lcl}
f_{Y|X}(y_m|\nu_m)
\label{Eq4a}\label{EqDensities}
& = & \frac{1}{\pi \sigma^2}\exp \left( - \frac{\left| y_m - \nu_m \right|^2}{\sigma^2} \right). 
 \end{IEEEeqnarray}
\end{subequations}
At the EH, the RF signal $Z(t)$ in~\eqref{Eq7b} is not down-converted or filtered. See for instance~\cite{7547357,8115220}. 


\subsection{Information and Energy Transmission} \label{SubsecInformationEnergyTransmission}
Let $M$ be the cardinality of the set of symbols at the output of the information source.
To accomplish the tasks of information and energy transmission, the transmitter makes use of an $(n,M,\mathcal{X})$-code defined as follows. 
\begin{definition}[$(n,M,\mathcal{X})$-code] \label{DefNmCode}
An $(n,M,\mathcal{X})$-code, with $\mathcal{X}$ in~\eqref{EqCIsymbols}, is a set of pairs
\begin{equation}
    \left \lbrace ({\boldsymbol u}(1),\mathcal{D}_1), ({\boldsymbol u}(2),\mathcal{D}_2), \ldots, ({\boldsymbol u}(M),\mathcal{D}_M)\right \rbrace,
\end{equation}
such that for all $(i,j) \in \{1,2, \ldots, M\}^2$, with $i\neq j$,
\begin{subequations}\label{EqCodeProperties}
\begin{align}
\label{eq:u_i}  &{\boldsymbol u}(i) = (u_1(i), u_2(i), \ldots, u_n(i)) \in \mathcal{X}^n, \\
&\mathcal{D}_i \subseteq \mathds{C}^n,  \mbox{ and } \\
&\mathcal{D}_i \cap \mathcal{D}_j = \emptyset.
\end{align}
\end{subequations}
\end{definition}
Note that Definition~\ref{DefNmCode} specifies the exact set of channel input symbols via the set $\mathcal{X}$.
Thus, constraints on the average power or peak-amplitude power are vacuous as those are fixed by the definition of $\mathcal{X}$.  

The associated encoding, decoding and energy harvesting operations for a given $(n,M,\mathcal{X})$-code are defined as follows.
\subsubsection{Encoding and decoding} 
Assume that the transmitter uses the $(n,M,\mathcal{X})$-code
\begin{equation} \label{Eqnm_code}
    \mathscr{C} \triangleq \{({\boldsymbol u}(1),\mathcal{D}_1), ({\boldsymbol u}(2),\mathcal{D}_2), \ldots, ({\boldsymbol u}(M),\mathcal{D}_M)\},
\end{equation}
that satisfies~\eqref{EqCodeProperties}.

To transmit message index $i$, with $i \in \lbrace 1,2, \ldots, M \rbrace$, the transmitter uses the codeword ${\boldsymbol u}(i)= (u_1(i), u_2(i), \ldots, u_n(i))$. That is, at channel use $m$, with $m \in \lbrace 1,2, \ldots, n\rbrace$, the transmitter inputs the RF signal corresponding to symbol $u_{m}(i)$ into the channel. At the end of $n$ channel uses, the IR observes a realization of the random vector ${\boldsymbol Y} = (Y_1, Y_2, \ldots, Y_n)^{\sf{T}}$ in~\eqref{EqChannelModel}.
Let $W$ be a random variable that represents the output of the information source. 
For all $i \inCountK{M}$, the probability of transmitting codeword ${\boldsymbol u}(i)$ is 
\begin{IEEEeqnarray}{rCl} \label{Eqprior}
P_W \left(i \right) = \frac{1}{M},
\end{IEEEeqnarray}
which represents a maximum entropy information source.
 The IR decides that message index $i$ was transmitted, if the following event takes place:
\begin{equation}
    {\boldsymbol Y} \in \mathcal{D}_i,
\end{equation}
with $\mathcal{D}_i$ in~\eqref{Eqnm_code}.
Therefore, the DEP associated with the transmission of message index $i$ is 
\begin{IEEEeqnarray}{rCl}
\label{EqDEPi}
    \gamma_i(\mathscr{C}) 
    &\triangleq& 1 - \int_{\mathcal{D}_i} f_{\boldsymbol{Y}|\boldsymbol{X}}(\boldsymbol{y}|\boldsymbol{u}(i)) \mathrm{d}\boldsymbol{y},
\end{IEEEeqnarray}
with the conditional density $f_{\boldsymbol{Y}|\boldsymbol{X}}$ defined in \eqref{EqYXdistribution}.
Alternatively, the average DEP for code $\mathscr{C}$ is
\begin{IEEEeqnarray}{rCl} 
    \label{eq:gamma} 
    \gamma(\mathscr{C}) &\triangleq& \frac{1}{M}\sum_{i = 1}^M \gamma_i(\mathscr{C}). 
\end{IEEEeqnarray}

The information transmission rate of any $(n,M,\mathcal{X})$-code $\mathscr{C}$, denoted by $R(\mathscr{C})$, satisfies
\begin{equation} \label{EqR}
    R(\mathscr{C}) \triangleq \frac{\log_2 M}{n},
\end{equation}
in bits per channel use. 

\subsubsection{Energy Harvesting}
The channel output observed at the EH while transmitting codeword $\boldsymbol{u}(i)$, with $i \inCountK{M}$, is denoted by $Z_{i}(t)$, with $t \in [0,nT]$. From~\eqref{EqXrf} and~\eqref{Eq7b}, such a signal $Z_{i}(t)$ is 
\begin{IEEEeqnarray}{rcl} 
\label{Eq37} Z_{i}(t) &=&
\Re \Big(\sqrt{2} \sum_{m=1}^n u_m(i) \\
&&\sinc \left( f_w (t - (m-1)T) \right) \exp \left(\mathrm{i} 2 \pi f_c t \right) \Big) + N_2(t), \nonumber 
\end{IEEEeqnarray}
where, for all $m \inCountK{n}$, the complex $u_m(i)$ is the $m$\ts{th} symbol of the codeword $\boldsymbol{u}(i)$ in~\eqref{eq:u_i} and the random variable $N_2(t)$ is a real Gaussian random variable with zero mean and variance $\sigma^2$.
For all $i \inCountK{M}$ and $t \in [0,nT]$, denote by $x_{i}(t)$ the signal 
\begin{IEEEeqnarray}{l} \label{Eq39}
 x_{i}(t) = \\
 \Re \Big(\sqrt{2} \sum_{m=1}^n u_m(i) \sinc \left( f_w (t - (m-1)T) \right) \exp \left(\mathrm{i} 2 \pi f_c t \right) \Big). \nonumber
 \end{IEEEeqnarray}
 From~\eqref{Eq37} and~\eqref{Eq39} it follows that, 
 \begin{IEEEeqnarray}{rCl} 
Z_{i}(t) &=& x_{i}(t) + N_2(t). \label{Eq38}
\end{IEEEeqnarray}

Thus, for all $t \in [0,nT]$, the channel output $Z_i(t)$ in~\eqref{Eq38} is a real Gaussian random variable with mean $x_i(t)$ and variance~$\sigma^2$.
 
The energy harvesting model used in this work is the non-linear model introduced in~\cite{8115220} and~\cite{7547357}. This model states that the energy harvested from an RF signal is proportional to the DC component of the second and fourth powers of the signal. Hence, for all $i\inCountK{M}$, the expected energy harvested from the channel output $Z_i(t)$ in \eqref{Eq37} during the whole duration of the source symbol $i$ is 
\begin{IEEEeqnarray}{rCl} \label{Eqei}
e_i &\triangleq& k_1 \sum_{m=1}^n \left| u_m(i) \right|^2 + k_2 \sum_{m=1}^n \left| u_m(i) \right|^4, 
\end{IEEEeqnarray}
where 
the constants $k_1$ and $k_2$ are positive, with $k_1 = 0.0034$ and $k_2 = 0.3829$~\cite{8115220}; and for all $m \inCountK{n}$, the complex number $u_m(i)$ is in~\eqref{eq:u_i}. 

The energy harvested during the transmission of a codeword, i.e., during $n$ channel uses,  is a random variable denoted by~$E$.
The probability of harvesting energy $e$ given that codeword $\boldsymbol{u}(i)$ was transmitted is given by
\begin{IEEEeqnarray}{rCl} \label{Eq34}
P_{E | W} \left(e|i \right) = \mathds{1}_{\lbrace e = e_i \rbrace}, 
\end{IEEEeqnarray}
where the random variable  $W$ possesses the probability mass function defined in~\eqref{Eqprior}; and the energy level $e_i$ is defined in~\eqref{Eqei}.
The equality in~\eqref{Eq34} follows from the fact that once the transmitted codeword is known to be $\boldsymbol{u}(i)$, the energy harvested is a deterministic quantity $e_i$, as defined in~\eqref{Eqei}.
This is essentially due to the fact that the noise $N_2(t)$ in \eqref{Eq38} does not contribute any significant amount of energy. The signal from which energy is harvested is $x_{i}(t)$.
Therefore, the probability mass function of the harvested energy, given that codeword $\boldsymbol{u}(i)$ was transmitted is a point mass concentrated at $e_i$.

The probability mass function of the random variable $E$, denoted by $P_E$ is given by the following:
\begin{IEEEeqnarray}{rCl} \label{Eq35}
P_{E} \left(e\right) &=& \sum_{i=1}^M P_{E | W} \left(e|i \right) P_W \left(i \right) \\
&=& \frac{1}{M} \sum_{i=1}^M \mathds{1}_{\lbrace e = e_i \rbrace}. \label{Eq35b}
\end{IEEEeqnarray}
Let $B \in \reals$ be the energy required at the EH. Then, using~\eqref{Eq35}, the EOP associated with code $\mathscr{C}$ is the probability of transmitting a message $i \in \lbrace 1, 2, \ldots, M\rbrace$, such that $e_i$ in \eqref{Eqei} satisfies $e_i < B$. That is,  
\begin{IEEEeqnarray}{rCl} \label{eq:theta_def}
 \theta(\mathscr{C},B) &\triangleq&
 \sum_{i \in \lbrace j \inCountK{M}: e_j < B \rbrace}  P_E (e_i) \\
 &=& \sum_{i \in \lbrace j \inCountK{M}: e_j < B \rbrace} \frac{1}{M} \sum_{i=1}^M \mathds{1}_{\lbrace e = e_i \rbrace} \label{Eq37b} \\
  &=& \frac{1}{M} \Big| \lbrace i \inCountK{M}: e_i<B \rbrace \Big| \\
  & = & \frac{1}{M} \sum_{i=1}^M \mathds{1}_{\lbrace e_i < B \rbrace}, \label{Eq25}
 \end{IEEEeqnarray}
 where the equality in~\eqref{Eq37b} follows from~\eqref{Eq35b}. 

Definition~\ref{DefNmCode} can now be refined to include the DEP and the EOP as follows.
\begin{definition}[$(n,M,\mathcal{X},\epsilon,B,\delta)$-code] \label{def:nmed_code}
An $(n,M,\mathcal{X})$-code~$\mathscr{C}$ is said to be an $(n,M,\mathcal{X},\epsilon,B,\delta)$-code for the random transformation in~\eqref{EqYXdistribution}  if  the following hold:
\begin{IEEEeqnarray}{rCl}
\label{EqGammaUpperbound}
\gamma(\mathscr{C}) &\leq& \epsilon, \mbox{and} \\
\label{eq:delta}
    \theta(\mathscr{C},B) &\leq& \delta,
\end{IEEEeqnarray}
where $\gamma(\mathscr{C})$ and $\theta(\mathscr{C},B)$ are defined in \eqref{eq:gamma} and \eqref{eq:theta_def}, respectively.
\end{definition}

The $(n,M,\mathcal{X},\epsilon,B,\delta)$-codes thus defined can be employed for transmitting at most $B$ energy units in $n$ channel uses and at most $\frac{\log_2 M}{n}$ bits per channel use, while ensuring that the resulting DEP is smaller than $\epsilon$ and the EOP is smaller than~$\delta$.

The primary objective of what follows in Section~\ref{sec:results} is a characterization of the trade-offs between the various information and energy parameters of $(n,M,\mathcal{X},\epsilon,B,\delta)$-codes.
These trade-offs are presented in the form of inequalities involving the code parameters $n,M,\mathcal{X},\epsilon,B$ and $\delta$.
This characterization serves two main purposes.
First, these inequalities define what tuples $(n,M,\mathcal{X},\epsilon,B,\delta)$ are ``\emph{possible}".
In other words, these inequalities define the necessary conditions that an $(n,M,\mathcal{X},\epsilon,B,\delta)$-code must satisfy.
Second, the dependence of $M, \epsilon, B, \delta$ on common parameters precisely quantifies the trade-offs between the objectives of information and energy transmission.


%
\section{Necessary Conditions} \label{sec:results}
This section introduces some  inequalities involving the information and energy rates, the DEP and the EOP that any $(n,M,\mathcal{X},\epsilon,B,\delta)$-code satisfies. From this perspective, these inequalities are recognized to be necessary conditions for SIET.   
Furthermore, these inequalities quantify the mutual dependence between the information rate, the energy rate, the DEP and the EOP of $(n,M,\mathcal{X},\epsilon,B,\delta)$-codes, which in turn, provides insights into the feasibility of SIET for specific reliability guarantees.
For a specific $(n,M,\mathcal{X},\epsilon,B,\delta)$-code, e.g., the code $\mathscr{C}$ in~\eqref{Eqnm_code}, this mutual dependence is presented in terms of the \emph{type} induced by each of the codewords of such a code. 
The type induced by the codeword $\boldsymbol{u}(i)$, with $i \inCountK{M}$, is a probability mass function, denoted by $P_{\boldsymbol{u}(i)}$, such that for all $x \in \mathcal{X}$,
\begin{equation} \label{eq:u_measure}
    P_{\boldsymbol{u}(i)}(x) \triangleq \frac{1}{n} \sum_{m=1}^n \mathds{1}_{\{u_m(i) = x\}}.
\end{equation}
The type induced by all the codewords in $\mathscr{C}$ is also a probability mass function on the set $\mathcal{X}$ in~\eqref{EqCIsymbols}. This type is denoted by  $P_{\mathscr{C}}$ and  for all $x \in \mathcal{X}$,
\begin{equation} \label{eq:p_bar}
    P_{\mathscr{C}}(x) \triangleq \frac{1}{M} \sum_{i=1}^M P_{\boldsymbol{u}(i)}(x).
\end{equation}
For the ease of presentation, the main results are presented for a specific class of $(n,M,\mathcal{X},\epsilon,B,\delta)$-codes, referred to as \emph{constant composition} codes. Nonetheless, the results presented in this section can be easily extended to the whole class of $(n,M,\mathcal{X},\epsilon,B,\delta)$-codes. 
\begin{definition}[Constant Composition Codes]\label{DefHC}
An $(n,M,\mathcal{X})$-code $\mathscr{C}$ is said to be a constant composition code if for all $i \inCountK{M}$ and for all $x \in \mathcal{X}$, it holds that
\begin{equation}\label{EqHomogeneousCodes}
 P_{\boldsymbol{u}(i)}(x) =  P_{\mathscr{C}}(x),
\end{equation}
where, $P_{\boldsymbol{u}(i)}$ and $P_{\mathscr{C}}$ are the types defined in~\eqref{eq:u_measure} and~\eqref{eq:p_bar}, respectively.
\end{definition}
Constant composition codes are $(n,M,\mathcal{X})$-codes in which  a given channel input symbol is used the same number of times in all the codewords.
As a consequence, for a constant composition code $\mathscr{C}$, the amount of  energy $e_i$ in \eqref{Eqei} harvested while transmitting the codeword $i \in \lbrace 1,2, \ldots, M \rbrace$, is invariant with respect to  $i$. That is, for all $i \inCountK{M}$,
 \begin{IEEEeqnarray}{rCl} 
 \label{EqeiHomogeneous}
 e_i = e_\mathscr{C},
  \end{IEEEeqnarray}
and the positive real $e_\mathscr{C}$ satisfies  
 \begin{IEEEeqnarray}{rCl} \label{EqeiHomogeneousb}
e_\mathscr{C} \triangleq k_1 \sum_{x \in \mathcal{X}} n P_{\mathscr{C}} \left( x \right) \left| x \right|^2 + k_2 \sum_{x \in \mathcal{X}}n P_{\mathscr{C}} \left( x \right) \left| x \right|^4,
 \end{IEEEeqnarray}
where the constants $k_1$ and $k_2$ are the same as in \eqref{Eqei};  $\mathcal{X}$ in~\eqref{EqCIsymbols}; and $P_{\mathscr{C}}$ in~\eqref{eq:p_bar}.

The following theorem introduces a first caracterization of the dependences among the performance metrics of an $(n,M,\mathcal{X},\epsilon,B,\delta)$-code.
\begin{theorem} \label{TheoremImpossibility}
If a given constant composition $(n,M,\mathcal{X})$-code,~$\mathscr{C}$, with $\mathcal{X}$ in \eqref{EqCIsymbols}, is an $(n,M,\mathcal{X},\epsilon,B,\delta)$-code for the random transformation in~\eqref{EqYXdistribution}, then the following conditions simultaneously hold:
\begin{subequations}
\label{EqImpossibility}
\begin{IEEEeqnarray}{l}
 \label{EqeImp}
 \epsilon \geq (M-1) \mathrm{Q} \left( \sqrt{\frac{\sum_{\ell=1}^L n P_{\mathscr{C}}(x^{(\ell)}) \left| x^{(\ell)} - \bar{x}^{(\ell)} \right|^2}{2 \sigma^2}} \right), \\
\label{EqRbound}
R(\mathscr{C}) \leq \frac{1}{n} \log_2 \left( \frac{n!}{\prod_{\ell=1}^L (nP_{\mathscr{C}}(x^{(\ell)}))!}\right), \\
 \label{EqdImp}
 \delta \geq \mathds{1}_{\lbrace e_\mathscr{C} < B \rbrace}, \\
 \label{EqBImp}
 B \leq e_\mathscr{C},
\end{IEEEeqnarray}
\end{subequations}
where, the information rate $R(\mathscr{C})$ is defined in \eqref{EqRbound}; the function $Q$ is defined in \eqref{DefQfunc}; 
for all $\ell \inCountK{L}$, the complex $\bar{x}^{(\ell)} \in \mathcal{X}$ is \begin{IEEEeqnarray}{rCl} \label{EqNeighbor}
\bar{x}^{(\ell)} \in \arg\max_{x \in \mathcal{X} \setminus \lbrace x^{(\ell)} \rbrace} \left|x^{(\ell)} - x \right|;
\end{IEEEeqnarray} 
$P_{\mathscr{C}}$ is the type defined in~\eqref{eq:p_bar}; and 
$e_\mathscr{C}$ is defined in~\eqref{EqeiHomogeneousb}.
\end{theorem}
\begin{proof}
The proofs are presented in Appendix~\ref{AppendixA}.
\end{proof}
Theorem~\ref{TheoremImpossibility} introduces the fundamental limits 
on the main performance metrics, namely the information transmission rate, energy transmission rate, DEP, and EOP, of all constant composition $(n,M,\mathcal{X})$-codes. 
Essentially, if a given constant composition  $(n,M,\mathcal{X})$-code $\mathscr{C}$ is an $(n,M,\mathcal{X}, \epsilon, B,\delta)$-code, then 
the DEP $\epsilon$ cannot be smaller that the right-hand side of \eqref{EqeImp};
the information transmission rate $R(\mathscr{C})$ in \eqref{EqR} cannot be bigger than the right-hand side of \eqref{EqRbound};
the EOP $\delta$ cannot be smaller than the right-hand side of \eqref{EqdImp}; and
the energy transmitted cannot be bigger than the right-hand side of \eqref{EqBImp}.
These fundamental upper and lower bounds on the performance metrics depend on the type induced by the code $\mathscr{C}$.  
This observation reveals a profound conclusion: Two different  $(n,M,\mathcal{X})$-codes $\mathscr{C}_1$ and $\mathscr{C}_2$ of the form in \eqref{Eqnm_code}, might exhibit different codewords and different decoding regions, nonetheless, if they induce the same types, i.e., $P_{\mathscr{C}_1} = P_{\mathscr{C}_2}$, then their performance metrics to simultaneously transmit information and energy  are subject to identical constraints.
This reveals the relevance of the type $P_{\mathscr{C}}$ in \eqref{EqImpossibility} as the key parameter in code design for SIET.

In the remainder of this section, the necessary conditions on the information and energy transmission parameters of $(n,M,\mathcal{X},\epsilon,B,\delta)$-codes are discussed in further detail.

\subsection{DEP}
To minimize the DEP for the code $\mathscr{C}$, the decoding regions $\mathcal{D}_1$,  $\mathcal{D}_2$, $\ldots$, $\mathcal{D}_M$  are defined using the maximum a posteriori (MAP) decision rule~\cite[Chapter $21$]{lapidoth}, \ie, for all  $i \inCountK{M}$, 
 %
%
 \begin{IEEEeqnarray}{lcl}
\label{EqMAPregion}
\mathcal{D}_i 
 & = & \left\lbrace \boldsymbol{y} \in \complex : \forall \ j \inCountK{M},\frac{ f_{\boldsymbol{Y}|\boldsymbol{X}}(\boldsymbol{y}|\boldsymbol{u}(i))}{f_{\boldsymbol{Y}|\boldsymbol{X}}(\boldsymbol{y}|\boldsymbol{u}(j))} \geq 1 \right\rbrace. \IEEEeqnarraynumspace
\end{IEEEeqnarray}
The lower bound on the DEP, $\epsilon$, in \eqref{EqeImp}, is independent of the decoding regions $\mathcal{D}_1$,  $\mathcal{D}_2$, $\ldots$, $\mathcal{D}_M$ of the code $\mathscr{C}$. 
This lower bound is the exact DEP of a code in which all the channel input symbols in $\mathcal{X}$ are equidistant from one another, such as in the case of binary phase-shift keying.
Hence, the lower bound on the DEP, $\epsilon$, in \eqref{EqeImp} is tight for such a code.  

%

\subsection{Energy Transmission Rate and EOP}
The lower and upper bounds in \eqref{EqdImp} and \eqref{EqBImp}, which might appear redundant with respect to each other in the case of constant composition codes, are more informative in the general case as discussed below.
The following lemma provides a lower bound on $\delta$ for the general case of $(n,M,\mathcal{X},\epsilon,B,\delta)$-codes. 
\begin{lemma} \label{LemmaB}
If a given $(n,M,\mathcal{X})$-code,~$\mathscr{C}$, with $\mathcal{X}$ in \eqref{EqCIsymbols}, is an $(n,M,\mathcal{X},\epsilon,B,\delta)$-code for the random transformation in~\eqref{EqYXdistribution}, then the following holds for $\delta$ in~\eqref{eq:delta}:
\begin{IEEEeqnarray}{rCl}
 \label{eq:B_bound_lemma}
\delta \geq \frac{1}{M} \sum_{i=1}^M \mathds{1}_{\lbrace e_i < B \rbrace},
\end{IEEEeqnarray}
where, for all $i \inCountK{M}$, the real $e_i \in [0,\infty)$ is in~\eqref{Eqei}.
\end{lemma}
\begin{proof}
The result follows from~\eqref{Eq25} and~\eqref{eq:delta}.
\end{proof}
In Lemma~\ref{LemmaB}, the lower bound on the EOP, $\delta$, is the probability of the set of symbols that deliver energy that is smaller than the required value of $B$. That is, $\delta \geqslant P_{W} \left( \left\lbrace i \inCountK{M}: e_i < B \right\rbrace \right)$, with $P_{W}$ in \eqref{Eqprior}. Hence, if there exists at most one codeword $\boldsymbol{u}(i)$, for some $i \inCountK{M}$, such that $e_i < B$, then, the EOP is bounded away from zero. More specifically, $\delta \geqslant \frac{1}{M}$. If there are $\rho$ codewords $\boldsymbol{u}(\ell_1)$, $\boldsymbol{u}(\ell_2)$, $\cdots$, $\boldsymbol{u}(\ell_{\rho})$, such that for all $j \inCountK{\rho}$, $e_{\ell_j} < B$, then, $\delta \geqslant \frac{\rho}{M}$.
If all the codewords $\boldsymbol{u}(1)$, $\boldsymbol{u}(2)$, $\ldots$, $\boldsymbol{u}(M)$ are such that $\max_{i} e_{i} < B$, then the EOP is equal to one. That is, it is impossible to guarantee that~$B$ units of energy will be harvested at the EH. 
This reveals the entanglement between $B$ and $\delta$. 

To quantify the relationship between $B$ and $\delta$ , denote by $M' \leq M$, the number of unique values in the vector $\left(e_1, e_2, \ldots, e_M \right)^{\sf{T}}$ with $e_i$ in~\eqref{Eqei}. 
These $M'$ unique energy levels are represented by $\left\lbrace \bar{e}_1, \bar{e}_2, \ldots, \bar{e}_{M'} \right\rbrace \subset \reals$. 
More specifically, for all $i \inCountK{M}$, there exists $j \inCountK{M'}$ such that $e_i = \bar{e}_j$. 

Assume without loss of generality that the following holds:
\begin{IEEEeqnarray}{rCl} \label{EqUniqueLevels}
0 < \bar{e}_1 < \bar{e}_2 < \ldots < \bar{e}_{M'}.
\end{IEEEeqnarray}
Using this notation, the following lemma provides an upper bound on the energy requirement $B$ for any $(n,M,\mathcal{X},\epsilon,B,\delta)$-code as a function of the EOP $\delta$. 
\begin{lemma} \label{LemmaBnew}
Given an $(n,M,\mathcal{X},\epsilon,B,\delta)$-code $\mathscr{C}$ for the random transformation in~\eqref{EqYXdistribution}, with $\mathcal{X}$ in \eqref{EqCIsymbols}, let $j^{+} \in \ints$ be defined as follows.
\begin{IEEEeqnarray}{rcl}
j^{+} & \triangleq & \min \left\lbrace j \in \lbrace 1, \ldots M' \rbrace: \delta \leqslant \frac{\displaystyle\sum_{k=1}^{j}  \displaystyle\sum_{i=1}^M \mathds{1}_{\left\lbrace e_i = \bar{e}_k \right\rbrace }}{M} \right\rbrace,
\end{IEEEeqnarray}
where, the positive integer $M'$ and the reals $\bar{e}_1$, $\bar{e}_2$, $\ldots$, $\bar{e}_{M'}$ are in~\eqref{EqUniqueLevels}.
Then, the following holds for the energy requirement $B$:
\begin{IEEEeqnarray}{rCl} \label{EqConverseB}
B \leq \bar{e}_{j^{+}} 
\end{IEEEeqnarray}
\end{lemma}
\begin{proof}
    The proof is given in Appendix~\ref{AppendixE}.
\end{proof}
\subsection{Information Transmission Rate}
Though the bound on $R(\mathscr{C})$ in~\eqref{EqRbound} is tight, it can quickly become computationally infeasible as the block-length $n$ increases. 
The following theorem provides a tractable approximation of this bound. 
 
\begin{theorem} \label{CorRUpperBoundRelax} If a given constant composition $(n,M,\mathcal{X})$-code,~$\mathscr{C}$, with $\mathcal{X}$ in \eqref{EqCIsymbols}, is an $(n,M,\mathcal{X},\epsilon,B,\delta)$-code for the random transformation in~\eqref{EqYXdistribution}, then the following holds for  $R(\mathscr{C})$ in~\eqref{EqR}:
\begin{IEEEeqnarray}{l} \label{EqCorRUpperBoundRelax}
R(\mathscr{C}) \leq 
H\left( P_{\mathscr{C}} \right)  + \frac{1}{n^2} \left( \frac{1}{12} - \sum_{\ell=1}^L \frac{1}{12 P_{\mathscr{C}}(x^{(\ell)}) +1}\right) 
  \\
 + \frac{1}{n} \left(\log \left( \sqrt{2\pi}\right)   - \sum_{\ell=1}^L \log\sqrt{2\pi P_{\mathscr{C}}(x^{(\ell)})} \right) - \frac{\log n}{n} \left( \frac{L-1}{2} \right), \nonumber 
\end{IEEEeqnarray}
where, $P_{\mathscr{C}}$ is the type defined in~\eqref{eq:p_bar};  and $L$ is the number of channel input symbols in~\eqref{EqCIsymbols}.
\end{theorem}
\begin{proof}
    The proof is given in Appendix~\ref{AppendixC}.
\end{proof}

 Note that all terms in~\eqref{EqCorRUpperBoundRelax}, except the entropy $H\left( P_{\mathscr{C}} \right)$, vanish with the block-length $n$.
This implies the well known information theoretic result that information rate is essentially constrained by the entropy of the channel input symbols. In particular, note that  $H\left( P_{\mathscr{C}}\right) \leqslant \log_2 L$.
Furthermore, this bound holds with equality for a code $\mathscr{C}$ when the type $P_{\mathscr{C}}$ is uniform, \ie, all the symbols from the set of channel inputs $\mathcal{X}$ are used with the same frequency in $\mathscr{C}$.
It is important to note that, even though it may not be explicitly clear in Theorem~\ref{TheoremImpossibility} and Theorem~\ref{CorRUpperBoundRelax}, the inequalities on the information rate and the DEP are in fact dependent on each other. Each of the inequalities are a function of the type $P_{\mathscr{C}}$, which defines the dependence between these parameters. 

\begin{remark}
The value $\frac{n!}{\prod_{\ell=1}^L (nP_{\mathscr{C}}(x^{(\ell)}))!}$ in \eqref{EqRbound} is the exact number of different codewords of length $n$, while exhibiting an average type $P_{\mathscr{C}}$.  
Hence, codes such that every channel input symbol is used the same number of times among all the codewords are less constrained in terms of the information rate.
This is the case in which  $P_{\mathscr{C}}$ is a uniform distribution. Alternatively, using a uniform type $P_{\mathscr{C}}$ might reduce the energy transmission rate significantly. For instance, assume that the set of channel input symbols is such that for at least one pair $(x_1,x_2) \in \mathcal{X}^2$, with $\mathcal{X}$ in~\eqref{EqCIsymbols}, it holds that $\left| x_{1} \right| < \left| x_{2} \right|$. Then, from~\eqref{EqeiHomogeneousb}, it is clear that using the symbol $x_{1}$ equally often as $x_{2}$ constraints the value of $e_\mathscr{C}$ and hence, the energy $B$. Alternatively, codes that exhibit the largest energy rates are those in which the symbols that have the largest magnitude are used most often. This clearly deviates from the uniform distribution and thus, constrains the information rate~$R$. 
\end{remark}

\section{Code Construction for SIET} \label{SecAchievability} 
This section introduces a method for constructing codes for SIET that ensure certain reliability guarantees. 
The tuples of information rate, energy rate, DEP and EOP that can be achieved by the constructed codes are also characterized.
This is accomplished in two steps.
The first step is the construction of the constellation $\mathcal{X}$ defined in~\eqref{EqCIsymbols} and the construction of a particular $(n,M,\mathcal{X})$-code, whose codewords and decoding sets are subject to certain conditions such that such the code is  an $(n,M,\mathcal{X},\epsilon,B,\delta)$-code.

\subsection{Constellation Design}

Consider a constellation represented by the set $\mathcal{X}$ in~\eqref{EqCIsymbols} and assume that $\mathcal{X}$ is the union of  $C \in \ints$ sets of channel input symbols, which will be referred to as \emph{layers}. A layer is a set of symbols that have the same magnitude.
For all $c \in \lbrace 1,2, \ldots, C \rbrace$, denote by  $L_c \in \ints$ the number of symbols and by $A_c \in \reals^+$ the amplitude of the symbols in layer $c$.
The underlying assumption is that symbols in layer $c$ are equally spaced along the circle of radius $A_c$. This is typically the case in phase-shift keying (PSK) modulations~\cite[Chapter $16$]{lapidoth}.
 This is because, given the number of symbols in a layer, the symbols being equally spaced is the most favorable in terms of the DEP, as formally shown later in the proof of the main results of this section.

Let $\alpha_c \in [0,2\pi]$ denote the phase shift of the symbols in layer~$c$. The layer~$c$ is denoted as
\begin{IEEEeqnarray}{rCl} 
    \mathcal{U}(A_c,L_c,\alpha_c) &\triangleq& \lbrace x_c^{\left( 1 \right)},x_c^{\left( 2 \right)}, \ldots, x_c^{\left( L_c \right)} \rbrace \label{EqLayerCircle} \\
    &=& \Big\lbrace x \in \complex : x = A_c \exp\left(\mathrm{i} \left(\frac{2\pi}{L_c} \ell+\alpha_c\right)\right), \nonumber \\
    && \ell \in \left \lbrace 0, 1,2, \ldots, (L_c-1)\right \rbrace \Big\rbrace, 
\end{IEEEeqnarray}
where $\mathrm{i}$ is the complex unit. Hence, the set $\mathcal{X}$ is given by
\begin{IEEEeqnarray}{rCl} \label{EqConstellationCircle}
    \mathcal{X} & = & \bigcup_{c=1}^C \mathcal{U}(A_c,L_c,\alpha_c).
\end{IEEEeqnarray} 
The vector of the amplitudes $A_c$ in~\eqref{EqConstellationCircle} is denoted by
\begin{equation} \label{EqAcVector}
\boldsymbol{A} = \left( A_1, A_2, \ldots, A_C\right)^{\sf{T}},
\end{equation}
and without any loss of generality, assume that
\begin{equation} \label{EqAmplitudesOrder}
    A_1 > A_2 > \ldots > A_C.
\end{equation}

The vector of the number of symbols in each layer in~\eqref{EqConstellationCircle} is denoted by
\begin{equation} \label{EqLcVector}
\boldsymbol{L} = \left( L_1, L_2, \ldots, L_C\right)^{\sf{T}};
\end{equation}
and the vector of the phase shifts of the symbols in each layer in~\eqref{EqConstellationCircle} is denoted by
\begin{equation} \label{EqAlphaVector}
\boldsymbol{\alpha} = \left( \alpha_1, \alpha_2, \ldots, \alpha_C\right)^{\sf{T}}.
\end{equation}
The total number of symbols $L$ in~\eqref{EqCIsymbols} satisfies
\begin{equation} \label{EqLSum}
    L = \sum_{c=1}^C L_c.
\end{equation}
 
Using the design in \eqref{EqConstellationCircle} for the set of channel input symbols, $\mathcal{X}$ in~\eqref{EqCIsymbols}, the $(n,M,\mathcal{X},\epsilon,B,\delta)$-code $\mathscr{C}$ is constructed in two steps.
The first step consists of designing the codewords, while the second step designs the corresponding decoding regions.

\subsection{Codeword Design}\label{SecCodewordDesing}

Consider an $(n,M,\mathcal{X})$-code of the form in \eqref{Eqnm_code}, with $\mathcal{X}$ in~\eqref{EqConstellationCircle}, whose codewords are all different and are constructed satisfying a constraint on the number of times the symbols in the different layers are used. 
For all $c \inCountK{C}$, let $p_c$ be the empirical probability of symbols of the $c$\ts{th} layer. 
That is,
\begin{IEEEeqnarray}{rCl} \label{Eqpc}
    p_c &=& \frac{1}{Mn} \sum_{\ell = 1}^{L_c} \sum_{i=1}^M \sum_{m=1}^n \mathds{1}_{\{u_m(i) = x_c^{(\ell)}\}} \\
    &=& \frac{1}{M} \sum_{\ell = 1}^{L_c} \sum_{i=1}^M  P_{\boldsymbol{u}(i)}\left(x_c^{(\ell)} \right)\label{Eq115i}\\
    & = & \sum_{\ell = 1}^{L_c} P_{\mathscr{C}}\left(x_c^{(\ell)} \right),
\end{IEEEeqnarray}
where, 
for all $i \inCountK{M}$ and for all $m \inCountK{n}$, the complex $u_m(i)$ is the $m$\ts{th} channel input symbol of the codeword $\boldsymbol{u}(i)$ in \eqref{eq:u_i}; the types $P_{\boldsymbol{u}(i)}$ and $P_{\mathscr{C}}$ are defined in~\eqref{eq:u_measure} and in~\eqref{eq:p_bar}, respectively; and $x_c^{(\ell)}$ is in~\eqref{EqLayerCircle}.
The resulting probability vector is denoted by
\begin{equation} \label{Eqp}
    \boldsymbol{p} = \left( p_1,p_2, \ldots, p_C \right)^{\sf{T}},
\end{equation}
which is a free parameter in the construction of code $\mathscr{C}$. 
The constraint on the codewords of code $\mathscr{C}$ is that for all $c \inCountK{C}$ and for all $\ell \inCountK{L_c}$, the frequency with which the symbol $x_c^{(\ell)}$ appears in $\mathscr{C}$ satisfies
\begin{IEEEeqnarray}{rCl}
\label{EqTypeCircle} 
P_{\mathscr{C}}(x_c^{(\ell)})     &=& \frac{p_c}{L_c}.
\end{IEEEeqnarray}
That is, symbols within the same layer are used with the same frequency in~$\mathscr{C}$. 
The motivation for this assumption is that intuitively, while using two symbols within the same layer with different frequencies does not impact the energy transmission rate $B$ or EOP $\delta$, it does have an impact on the information transmission rate and DEP of $\mathscr{C}$. 

\subsection{Decoding Regions}\label{SecDRdesing}
The focus is now on the design of the decoding regions $\mathcal{D}_1$, $\mathcal{D}_2$, $\ldots$, $\mathcal{D}_M$ of the code $\mathscr{C}$, for the codewords $\boldsymbol{u}(1)$, $\boldsymbol{u}(2)$, $\ldots$, $\boldsymbol{u}(M)$ constructed above.  
For all $c \inCountK{C}$ and $\ell \inCountK{L_c}$, let the  set $\mathcal{G}_c^{(\ell)} \subseteq \complex$ be a circle of radius $r_c \in \reals^+$ centered at $x_c^{(\ell)}$, with $x_c^{(\ell)}$ in \eqref{EqLayerCircle}.
That is,
\begin{IEEEeqnarray}{l}
\label{EqDecodingCircleComplex} 
\mathcal{G}_c^{(\ell)} = \left\lbrace y \in \complex : \left| y - x_c^{\left(\ell \right)}\right|^2 \leq r_c^2 \right\rbrace.   
\end{IEEEeqnarray}
For all $i \inCountK{M}$, the decoding region $\mathcal{D}_i$ for codeword $\boldsymbol{u}(i)$ is built as follows:
\begin{equation} \label{EqDecodingCodeword}
    \mathcal{D}_i = \mathcal{D}_{i,1} \times \mathcal{D}_{i,2} \times \ldots \times \mathcal{D}_{i,n}, 
\end{equation}
where, for all $m \inCountK{n}$,
\begin{IEEEeqnarray}{rCl}
\mathcal{D}_{i,m} & = & \mathcal{G}_c^{(\ell)}, 
\end{IEEEeqnarray}
with $c$ and $\ell$ satisfying $u_m(i) = x_c^{(\ell)}$.

To ensure that the decoding regions $\mathcal{D}_1$, $\mathcal{D}_2$, $\ldots$, $\mathcal{D}_M$ are mutually disjoint, for all $c \inCountK{C-1}$, the amplitudes $A_c$ in~\eqref{EqLayerCircle} and the radii $r_c$ in~\eqref{EqDecodingCircleComplex} satisfy the following:
\begin{IEEEeqnarray}{rCl} 
\label{EqAmplitudesDifference}
A_{c+1} - A_{c} \geq r_{c+1} + r_{c}.
\end{IEEEeqnarray}

The vector of the radii in~\eqref{EqDecodingCircleComplex} is denoted by
\begin{equation} \label{EqRadiusVector}
\boldsymbol{r} = \left( r_1, r_2, \ldots, r_C\right)^{\sf{T}}.
\end{equation}

Note that there are multiple ways of choosing the parameters $\boldsymbol{A},\boldsymbol{\alpha},\boldsymbol{p}$ and $\boldsymbol{r}$ in the proposed construction to satisfy the same information rate, energy rate, DEP and EOP requirements.
Thus, the construction defines a family of codes
\begin{IEEEeqnarray}{rCl}
\label{EqFamily}
{\sf C} \left(C,\boldsymbol{A},\boldsymbol{L},\boldsymbol{\alpha},\boldsymbol{p},\boldsymbol{r} \right),
\end{IEEEeqnarray}
with the number of layers $C$ in~\eqref{EqConstellationCircle}; 
$\boldsymbol{A}$ in~\eqref{EqAcVector}; 
$\boldsymbol{L}$ in~\eqref{EqLcVector}; 
$\boldsymbol{\alpha}$ in~\eqref{EqAlphaVector}; 
$\boldsymbol{p}$ in~\eqref{Eqp}; 
and $\boldsymbol{r}$ in~\eqref{EqRadiusVector}.
The codewords  $\boldsymbol{u}(1)$, $\boldsymbol{u}(2)$, $\ldots$, $\boldsymbol{u}(M)$ and their respective decoding regions  $\mathcal{D}_1$, $\mathcal{D}_2$, $\ldots$, $\mathcal{D}_M$ completely define an  $(n,M,\mathcal{X})$-code within the family ${\sf C} \left(C,\boldsymbol{A},\boldsymbol{L},\boldsymbol{\alpha},\boldsymbol{p},\boldsymbol{r} \right)$.

\subsection{Achievable Tuples}\label{SecAch}

The following theorem characterizes the tuples of information rate, energy rate, DEP and EOP that can be achieved by codes within the family ${\sf C} \left(C,\boldsymbol{A},\boldsymbol{L},\boldsymbol{\alpha},\boldsymbol{p},\boldsymbol{r} \right)$ in~\eqref{EqFamily}.
\begin{theorem} \label{TheoremAchievableRegion}
A code $\mathscr{C}$ of the form in~\eqref{Eqnm_code} from the family ${\sf C} \left(C,\boldsymbol{A},\boldsymbol{L},\boldsymbol{\alpha},\boldsymbol{p},\boldsymbol{r} \right)$ in~\eqref{EqFamily} is an $\left(n,M,\mathcal{X},\epsilon,B,\delta \right)$-code  for the random transformation in~\eqref{EqYXdistribution},  if the following inequalities are simultaneously satisfied:
\begin{subequations}
\label{EqAchievability}
\begin{IEEEeqnarray}{l}
\label{EqAchieveDEP}
 \epsilon \geq 1 - \frac{1}{M}\sum_{i = 1}^M\prod_{c=1}^{C}  \left( 1-\exp \left(-\frac{r_c^2}{\sigma^2}\right) \right)^{n \sum_{\ell=1}^{L_c} P_{\boldsymbol{u}(i)}(x_c^{(\ell)})}, \\
 R(\mathscr{C})  \leq \log_2\left( \sum_{c=1}^C \left\lfloor \frac{\pi}{2\arcsin{\frac{r_c}{2A_c}}} \right\rfloor \right), \label{EqTheoremR} \\
 \label{EqTheoremrc} 
 \delta \geq  \frac{1}{M} \sum_{i=1}^M \mathds{1}_{\left\lbrace \substack{ \sum_{c=1}^C \sum_{\ell=1}^{L_c}  P_{\boldsymbol{u}(i)} \left( x_c^{(\ell)} \right) \left( k_1  A_c^2 +  k_2 A_c^4 \right) < \frac{B}{n}} \right\rbrace}, \label{EqTheoremB} \\
\label{EqTheoremB2}
B \leq \bar{e}_{j^{+}},
\end{IEEEeqnarray}
\end{subequations}
where, the information rate $R(\mathscr{C})$ is defined in \eqref{EqRbound}; the energy level $\bar{e}_{j^{+}}$ is defined in~\eqref{EqConverseB}; and the constants $k_1$ and $k_2$ are the same as in~\eqref{Eqei}; and the type $P_{\boldsymbol{u}(i)}$ is defined in~\eqref{eq:u_measure}.
\end{theorem}
\begin{proof}
The proofs are presented in Appendix~\ref{SecAchievableBounds}.
\end{proof}
Theorem~\ref{TheoremAchievableRegion} characterizes the conditions on the parameters $C,\boldsymbol{A},\boldsymbol{L},\boldsymbol{\alpha},\boldsymbol{p}$ and $\boldsymbol{r}$ such that all codes within the family ${\sf C} \left(C,\boldsymbol{A},\boldsymbol{L},\boldsymbol{\alpha},\boldsymbol{p},\boldsymbol{r} \right)$  in~\eqref{EqFamily} are $\left(n,M,\mathcal{X},\epsilon,B,\delta \right)$-codes.   
Hence, given some target information transmission rate $R$; energy requirement $B$; DEP $\epsilon$; and EOP $\delta$, Theorem~\ref{TheoremAchievableRegion} provides the conditions  on the parameters $\left(C,\boldsymbol{A},\boldsymbol{L},\boldsymbol{\alpha},\boldsymbol{p},\boldsymbol{r} \right)$ for the existence of an $\left(n,M,\mathcal{X},\epsilon,B,\delta \right)$-code. Such a code can be constructed by designing the codewords and the decoding regions as suggested in Section~\ref{SecCodewordDesing} and Section~\ref{SecDRdesing}, respectively.

\begin{remark} \label{Remark1}
It is essential to note the following about the choice of decoding sets in~\eqref{EqDecodingCircleComplex} and the derived bound on the achievable DEP in~\eqref{EqAchieveDEP}.
Firstly, the circular construction of decoding sets in~\eqref{EqDecodingCircleComplex} is a choice made in order to obtain closed form tractable expressions for the lower bound on the DEP in~\eqref{EqAchieveDEP}.
This is enabled by the inherent circular symmetry of the complex AWGN noise.
The bound in~\eqref{EqAchieveDEP} reveals the dependence of the DEP on the various parameters of the code such as the type.
This is instructive for studying the trade-offs between the parameters of the constructed family of codes.
Secondly, this choice of decoding sets provides guidelines for constructing the set of channel input symbols $\mathcal{X}$ while ensuring that the DEP does not exceed a required value which is controlled using the choice of the radii of the decoding sets in~\eqref{EqRadiusVector}.
\end{remark}

\section{Performance Analysis} 
\label{SecExamples}
The performance of the constructed family of codes can be evaluated by comparing the achievability bounds identified for the constructed family of codes ${\sf C} \left(C,\boldsymbol{A},\boldsymbol{L},\boldsymbol{\alpha},\boldsymbol{p},\boldsymbol{r} \right)$ in Theorem~\ref{TheoremAchievableRegion} with the necessary conditions in Theorem~\ref{TheoremImpossibility} that any $(n,M,\mathcal{X},\epsilon,B,\delta)$-code must satisfy.
This comparison is presented in the form of information-energy regions in Fig.~\ref{FigRegions} where the information transmission rates are plotted as a function of the energy transmission rates of the codes.
For this comparison, consider constant composition $(n,M,\mathcal{X},\epsilon,B,\delta)$-codes in the family ${\sf C} \left(C,\boldsymbol{A},\boldsymbol{L},\boldsymbol{\alpha},\boldsymbol{p},\boldsymbol{r} \right)$ in~\eqref{EqFamily} that employ the set of channel input symbols $\mathcal{X}$ of the form in~\eqref{EqConstellationCircle} with number of layers $C = 3$. The duration of the transmission in channel uses is $n = 80$. 
The radii of the decoding regions $r_c$ are assumed to be the same for all the layers i.e., for all $c \inCountK{C}$, the radius $r_c = r$ in~\eqref{EqDecodingCircleComplex}.
The amplitude of the first layer is $A_1 = 30$. Amplitudes of the second and third layers $A_2$ and $A_3$ are determined by $r$ to satisfy~\eqref{EqAmplitudesDifference}. The points on the curves in Fig.~\ref{FigRegions} are obtained by varying $\epsilon$ and the probability vector $\boldsymbol{p}$ in~\eqref{Eqp}.

The first important observation from Fig.~\ref{FigRegions} is that the achievability results for the codes constructed in this work match the necessary conditions identified in Theorem~\ref{TheoremImpossibility} except for the DEP $\epsilon$.
This is illustrated by the overlap of the necessary conditions and the achievable information-energy curves.
However, for the same information rate and energy pair, the DEP for the achievable curves is higher than the DEP bound identified by the necessary conditions. 
The sub-optimality in DEP is a result of the choice of circular decoding regions in~\eqref{EqDecodingCircleComplex} (See Remark~\ref{Remark1}). 

Fig.~\ref{FigRegions} also shows a clear trade-off between the information and energy transmission rates that can be simultaneously supported by a given code.
The maximum achievable information transmission rate is $R = 4.93$ bits/channel use. This $R$ is achieved by a code in which all the symbols in the set of channel inputs $\mathcal{X}$ are used with the same frequency. The maximum energy that can be delivered at $R = 4.93$ bits/channel use is $B = 3.2 \times 10^5$ energy units. This corresponds to the point $D_1$ in Fig.~\ref{FigRegions}.
The maximum achievable $B$ is $3.8 \times 10^{5}$ energy units. This is achieved by a code that exclusively uses the symbols in the first layer i.e., the probability vector $\boldsymbol{p}$ in~\eqref{Eqp} is $\boldsymbol{p} = (1,0,0)^{\sf{T}}$. The maximum $R$ that can be achieved at $B = 3.8 \times 10^{5}$ energy units is $R = 3.8$ bits/channel use. This corresponds to the point $D_2$ in Fig.~\ref{FigRegions}.
The curves between the points $D_1$ and $D_2$ in Fig.~\ref{FigRegions} illustrate the trade-off between the information and energy transmission.
As $B$ is increased from $3.2 \times 10^5$ energy units at point $D_1$, $R$ begins to decreases.
Similarly, as $R$ is increased from $3.8$ bits/channel use at point $D_2$, $B$ begins to decrease.
 \begin{figure}[htb]
  \centering
  \includegraphics[width=0.7\textwidth]{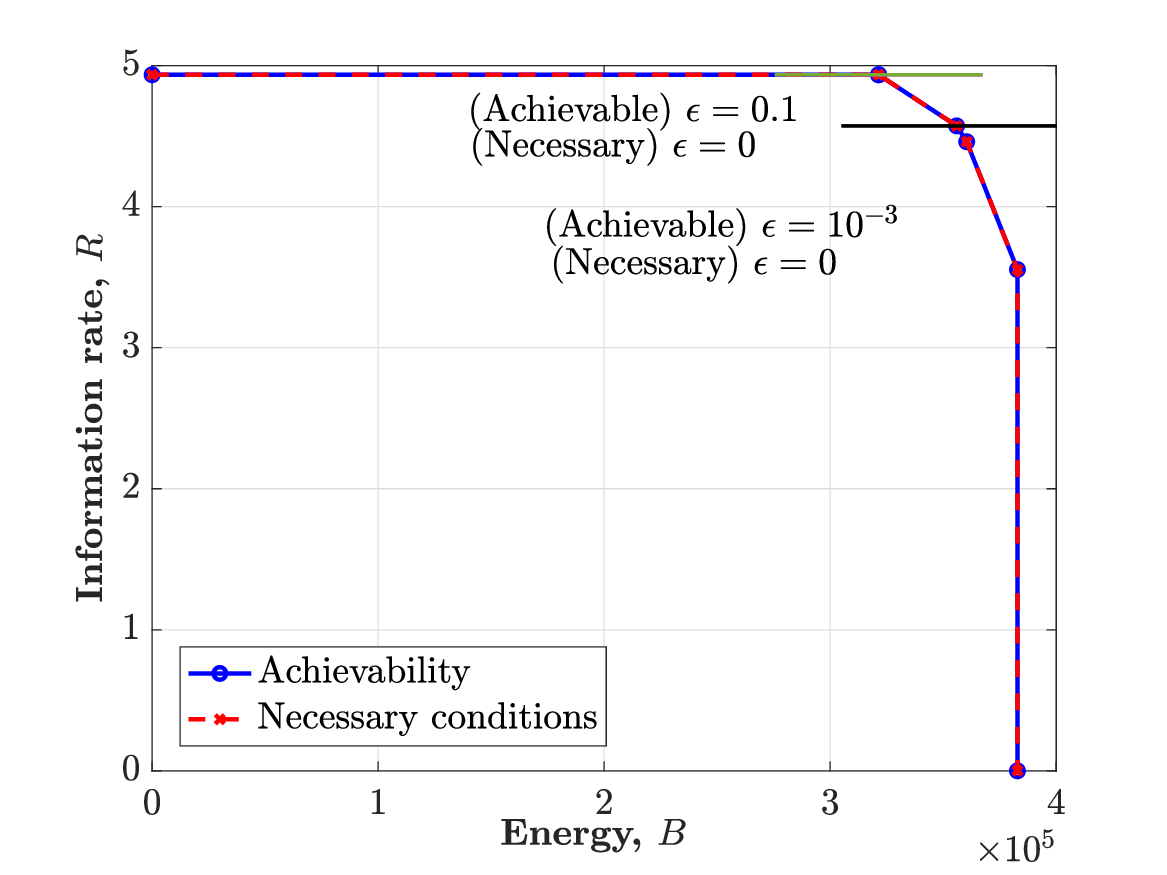}
  \caption{A comparison of the information-energy region defined by the necessary conditions (Theorem~\ref{TheoremImpossibility}) versus the achievable information-energy region (Theorem~\ref{TheoremAchievableRegion}) for constant composition codes in the family ${\sf C} \left(C,\boldsymbol{A},\boldsymbol{L},\boldsymbol{\alpha},\boldsymbol{p},\boldsymbol{r} \right)$.}
\label{FigRegions}
\end{figure}
\section{Discussion} \label{SecDiscussion} 
\subsubsection{Trade-offs} 
The results presented in this work reveal several interesting insights into the trade-offs between the information transmission rate $R$, the energy requirement $B$, the DEP $\epsilon$ and the EOP $\delta$. For instance, as shown by~\eqref{EqdImp}  and~\eqref{EqTheoremB}, $\delta$ increases as $B$ increases and vice versa. The consequence of this relationship is that a smaller $\delta$ can be achieved at the cost of lower $B$.
Similarly, for increasing $B$ in~\eqref{EqBImp} and~\eqref{EqTheoremB2}, a larger value of $\delta$ has to be tolerated.
These relationships also reveal that both $B$ and $\delta$ can be improved by a code that has higher values of $e_i$ in~\eqref{Eqei}. This is achieved by using the symbols with greater energy more frequently in the code. 

From the necessary condition in~\eqref{EqeImp}, it follows that $\epsilon$  decreases as the distance between the symbols in the set of channel inputs $\mathcal{X}$ in~\eqref{EqCIsymbols} increases. This implies that, a lower DEP can be achieved by increasing the distances between the channel input symbols. However, with the peak-amplitude constraint in place, increasing the distance between symbols implies that the number of symbols $L$ decreases.
This in turn, decreases the upper bound on the information rate $R$ in~\eqref{EqRbound}.
In~\eqref{EqTheoremrc}, $\epsilon$  decreases as a function of the radii $r_c$ of the decoding regions in~\eqref{EqDecodingCircleComplex}.
On the other hand, from~\eqref{EqTheoremR} it follows that the information rate $R$  increases as the radii $r_c$ decrease.
These trade-offs between $R$, $B$, $\epsilon$, and $\delta$ are more comprehensively illustrated in the following example.
\begin{figure}[t]
  \centering
  \includegraphics[width=0.7\textwidth]{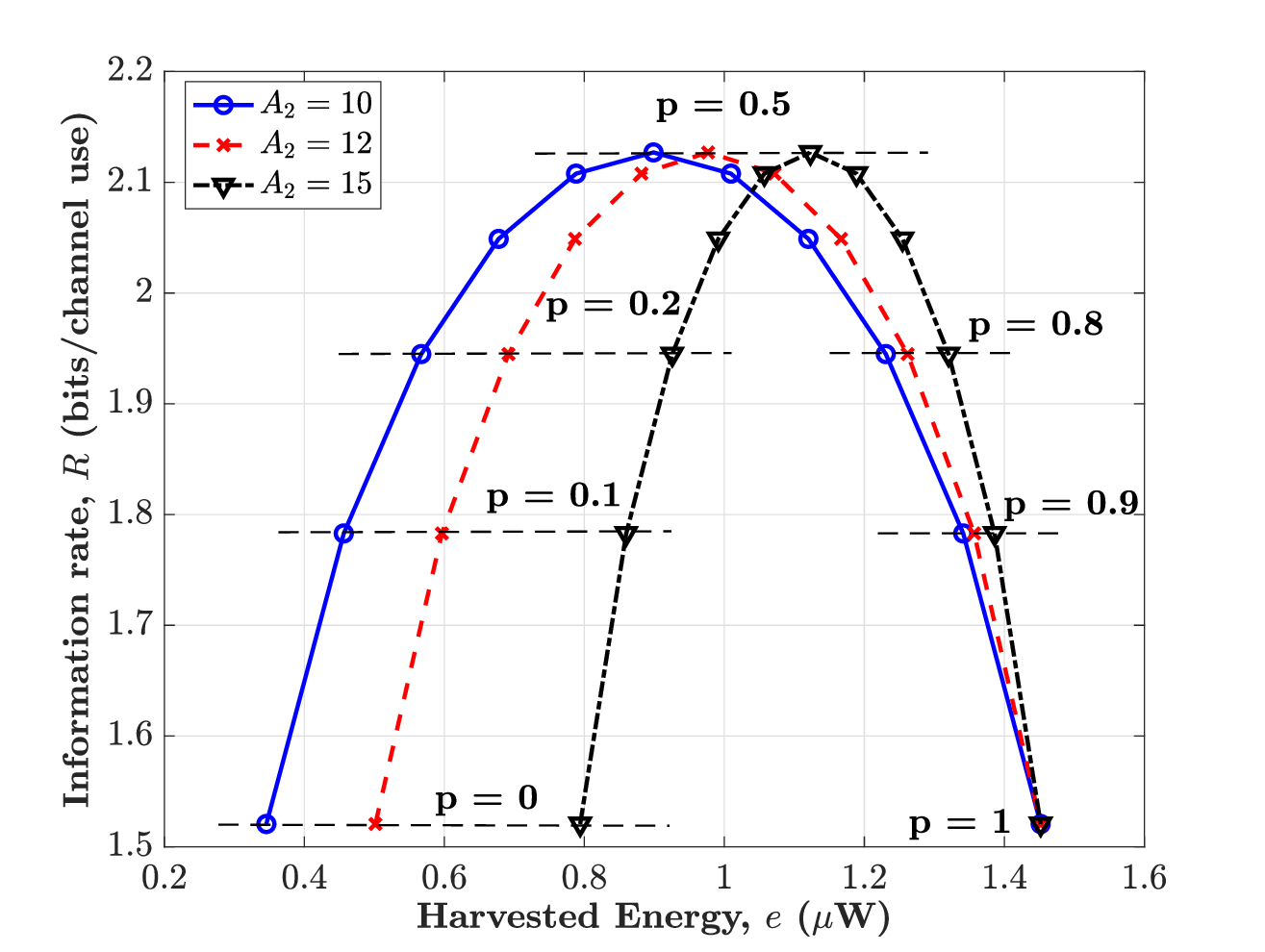}
\caption{Bounds on the information transmission rate $R$ in~\eqref{EqRbound} as a function of the harvested energy $e$ in~\eqref{Eq33}.}
\label{FigBR}
\end{figure}
%
%
%

Consider a family of constant composition $(n,M,\mathcal{X},\epsilon,B,\delta)$-codes $\mathscr{C}$ in the family ${\sf C} \left(C,\boldsymbol{A},\boldsymbol{L},\boldsymbol{\alpha},\boldsymbol{p},\boldsymbol{r} \right)$ in~\eqref{EqFamily}.
The set of channel input symbols $\mathcal{X}$ in \eqref{EqConstellationCircle} is composed of two layers with $5$ symbols in each layer, \ie, $C=2$ and $L_1 = L_2 = 5$. The radius of the first layer is $A_1 = 20$ millivolts and the radius of the second layer $A_2$ is varied to illustrate the trade-offs between the various parameters. The frequency with which symbols from the first layer appear in the code is $p_1 = p = 1-p_2$. That is, the vector $\boldsymbol{p}$ in~\eqref{Eqp} is given by
\begin{IEEEeqnarray}{l} \label{Eq71}
\boldsymbol{p} = \left( p, (1-p)\right)^{\sf{T}}.
 \end{IEEEeqnarray}
 The duration of the transmission in channel uses is $n = 100$. Since $\mathscr{C}$ is a constant composition code, from~\eqref{EqeiHomogeneous}, it holds that, for all $i \inCountK{M}$, 
 \begin{IEEEeqnarray}{rCl}
 e_i = e, \label{Eq33}
 \end{IEEEeqnarray}
where $e \in [0,\infty)$ is calculated as in~\eqref{EqeiHomogeneousb}. 
 
Fig.~\ref{FigBR} shows the trade-offs between the the information transmission rate $R$ in~\eqref{EqRbound} and the harvested energy $e$ in microwatts ($\mu W$) in~\eqref{Eq33} as a function of $p$ in~\eqref{Eq71}. Each curve in the figure is generated for some value of $A_2 < A_1$ by varying the value of $p \in [0,1]$. The following trade-offs can be observed from this figure. The harvested energy $e$ increases as $p$ increases. This is because higher $p$ corresponds to the symbols from the first layer $c=1$ which have higher energy (since $A_1 > A_2$) being used more frequently in $\mathscr{C}$.
For a fixed value of $A_2$ in Fig.~\ref{FigBR}, the information rate $R$ first increases and then decreases as a function of $e$ in~\eqref{Eq33}. For each of these curves, the maximum $R = 2.13$ bits/channel use corresponds to the uniform type, \ie, $p = 0.5$. For $p$ lesser or greater than $0.5$, the bound on $R$ decreases. Furthermore, the bounds on $R$ are independent of the values of $A_1$ and $A_2$. This is due to the fact that the information rate $R$ in~\eqref{EqRbound} is only a function of the type. 
The harvested energy $e$ also increases as $A_2$ increases. This is because, higher values of $A_2$ imply higher energy contained in the symbols in the second layer which in turn increases $e$. 

\subsubsection{Comparison with state of the art}
\begin{figure}[htb] 
  \centering
  \includegraphics[width=0.7\textwidth]{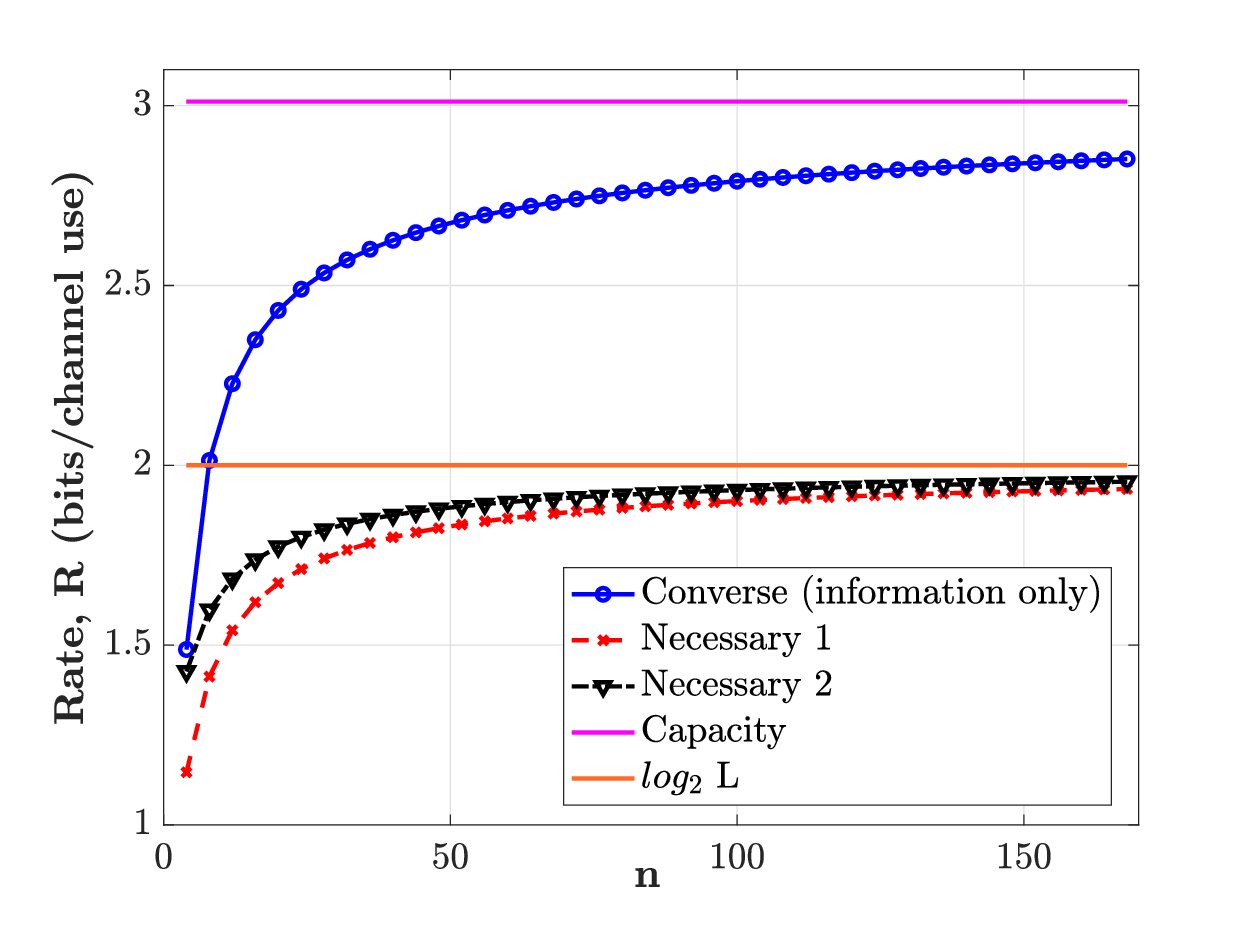}
  \caption{A comparison of the bounds on the information rate from the necessary conditions in~\eqref{EqRbound} (Necessary 1) and~\eqref{EqCorRUpperBoundRelax} (Necessary 2) for finite block-length SIET with the finite block-length converse in~\cite{polyanskiy2010channel} as a function of the block-length $n$.}
\label{FigCompare}
\end{figure}

Fig.~\ref{FigCompare} compares the bound on the information rate in~\eqref{EqRbound} (Necessary 1) and~\eqref{EqCorRUpperBoundRelax} (Necessary 2) with the finite block-length converse bound on the information rate in~\cite[Theorem 54]{polyanskiy2010channel} for a QPSK (\ie, $L=4$) constellation and signal to noise ratio (SNR) equal to $18$ dB.
It should be noted that the converse bounds in~\cite{polyanskiy2010channel} are concerned with information transmission alone. 
The following observations can be made from these plots. Firstly, the proposed bounds are very close to the finite block-length converse for small block-lengths.
As the block-length increases, the difference between the two bounds increases.
Secondly, the bounds in~\eqref{EqRbound} and~\eqref{EqCorRUpperBoundRelax} are very close. This shows that the approximation in~\eqref{EqCorRUpperBoundRelax} is tight, especially as the block-length increases. Finally, it should be noted that the bounds in~\eqref{EqRbound} and~\eqref{EqCorRUpperBoundRelax} are for the case of constant composition codes whereas the bound from~\cite{polyanskiy2010channel} has no such restriction. In fact, the information rate $R$ can be at most $\log_2 L$ in this case. 
As can be observed from Fig.~\ref{FigCompare}, the proposed bounds actually come very close to the maximum value of $\log_2 L$ as the block-length $n$ increases.

%
 \bibliographystyle{IEEEtran}
\bibliography{myrefs}

\begin{thebibliography}{10}
\providecommand{\url}[1]{#1}
\csname url@samestyle\endcsname
\providecommand{\newblock}{\relax}
\providecommand{\bibinfo}[2]{#2}
\providecommand{\BIBentrySTDinterwordspacing}{\spaceskip=0pt\relax}
\providecommand{\BIBentryALTinterwordstretchfactor}{4}
\providecommand{\BIBentryALTinterwordspacing}{\spaceskip=\fontdimen2\font plus
\BIBentryALTinterwordstretchfactor\fontdimen3\font minus
  \fontdimen4\font\relax}
\providecommand{\BIBforeignlanguage}[2]{{%
\expandafter\ifx\csname l@#1\endcsname\relax
\typeout{** WARNING: IEEEtran.bst: No hyphenation pattern has been}%
\typeout{** loaded for the language `#1'. Using the pattern for}%
\typeout{** the default language instead.}%
\else
\language=\csname l@#1\endcsname
\fi
#2}}
\providecommand{\BIBdecl}{\relax}
\BIBdecl

\bibitem{Tesla-Patent-1914}
N.~Tesla, \emph{{Apparatus for Transmitting Electrical Energy}}.\hskip 1em plus
  0.5em minus 0.4em\relax New York, NY: {U}nited {S}tates {P}atent {O}ffice,
  Dec. 1914, vol. US1119732 A.

\bibitem{GroverSahai}
P.~Grover and A.~Sahai, ``{Shannon Meets Tesla: Wireless Information and Power
  Transfer},'' in \emph{Proceedings of the {IEEE} International Symposium on
  Information Theory}, 2010, pp. 2363--2367.

\bibitem{Iqbal2024}
A.~Iqbal, A.~Smida, M.~Al-Hasan, I.~B. Mabrouk, and T.~A. Denidni, ``{Highly
  Isolated Wireless Power Transfer and Information Co-Delivery Using a
  Pacemaker Duplex Antenna},'' \emph{IEEE Transactions on Microwave Theory and
  Techniques}, p. 1–13, 2024.

\bibitem{varshney2008transporting}
L.~R. Varshney, ``{Transporting Information and Energy Simultaneously},'' in
  \emph{Proceedings of the {IEEE} International Symposium on Information Theory
  ({ISIT})}, Toronto, ON, Canada, Jul. 2008, pp. 1612--1616.

\bibitem{amor2016fundamental}
S.~B. Amor and S.~M. Perlaza, ``{Fundamental Limits of Simultaneous Energy and
  Information Transmission},'' in \emph{Proceedings of the International
  Conference on Telecommunications ({ICT})}, Thessaloniki, Greece, May 2016,
  pp. 1--5.

\bibitem{NizarItw2021}
N.~Khalfet and I.~Krikidis, ``{The Capacity of SWIPT Systems over
  Rayleigh-Fading Channels with HPA},'' in \emph{Proceedings of the {IEEE}
  Information Theory Workshop (ITW)}, Kanazawa, Japan, 2021, pp. 1--6.

\bibitem{amor2016feedback}
S.~B. Amor, S.~M. Perlaza, I.~Krikidis, and H.~V. Poor, ``{Feedback Enhances
  Simultaneous Energy and Information Transmission in Multiple Access
  Channels},'' in \emph{Proceedings of the {IEEE} International Symposium on
  Information Theory ({ISIT})}, Barcelona, Spain, Jul. 2016, pp. 1974--1978.

\bibitem{KhalfetGIC}
N.~Khalfet and S.~M. Perlaza, ``{Simultaneous Information and Energy
  Transmission in the Two-User {G}aussian Interference Channel},'' \emph{IEEE
  Journal on Selected Areas in Communications}, vol.~37, no.~1, pp. 156 --170,
  Jan. 2019.

\bibitem{8115220}
B.~Clerckx, ``{Wireless Information and Power Transfer: Nonlinearity, Waveform
  Design, and Rate-Energy Tradeoff},'' \emph{IEEE Transactions on Signal
  Processing}, vol.~66, no.~4, pp. 847--862, 2018.

\bibitem{amor2017feedback}
S.~B. Amor, S.~M. Perlaza, I.~Krikidis, and H.~V. Poor, ``{Feedback Enhances
  Simultaneous Wireless Information and Energy Transmission in Multiple Access
  Channels},'' \emph{IEEE Transactions on Information Theory}, vol.~63, no.~8,
  pp. 5244--5265, 2017.

\bibitem{perlaza2018simultaneous}
S.~M. Perlaza, A.~Tajer, and H.~V. Poor, ``{Simultaneous Information and Energy
  Transmission: {A} Finite Block-length Analysis},'' in \emph{Proceedings of
  the {IEEE} International Workshop on Signal Processing Advances in Wireless
  Communications (SPAWC)}, Kalamata, Greece, Jun. 2018, pp. 1--5.

\bibitem{khalfet2019ultra}
N.~Khalfet, S.~M. Perlaza, A.~Tajer, and H.~V. Poor, ``{On Ultra-reliable and
  Low Latency Simultaneous Information and Energy Transmission Systems},'' in
  \emph{Proceedings of the {IEEE} International Workshop on Signal Processing
  Advances in Wireless Communications (SPAWC)}, Cannes, France, Jul. 2019, pp.
  1--5.

\bibitem{zuhraITW}
S.~u. Zuhra, S.~M. Perlaza, and E.~Altman, ``{Simultaneous Information and
  Energy Transmission with Finite Constellations},'' in \emph{Proceedings of
  the {IEEE} Information Theory Workshop (ITW)}, Kanazawa, Japan, Oct. 2021,
  pp. 1--6.

\bibitem{zuhraISIT}
S.~u. Zuhra, S.~M. Perlaza, H.~V. Poor, and E.~Altman, ``{Achievable
  Information-Energy Region in the Finite Block-Length Regime with Finite
  Constellations},'' in \emph{Proceedings of the {IEEE} International Symposium
  on Information Theory (ISIT)}, Espoo, Finland, Jun. 2022, pp. 2106--2111.

\bibitem{10613795}
N.~Khalfet, C.~Psomas, S.~Chatzinotas, and I.~Krikidis, ``{Semantic
  Communications for Simultaneous Wireless Information and Power Transfer},''
  \emph{IEEE Transactions on Communications}, pp. 1--1, 2024.

\bibitem{varasteh2017wireless}
M.~Varasteh, B.~Rassouli, and B.~Clerckx, ``{Wireless Information and Power
  Transfer over an {AWGN} Channel: Nonlinearity and Asymmetric {Gaussian}
  Signaling},'' in \emph{Proceedings of the {IEEE} Information Theory Workshop
  (ITW)}, Kaohsiung, Taiwan, Nov. 2017, pp. 181--185.

\bibitem{7547357}
B.~Clerckx and E.~Bayguzina, ``{Waveform Design for Wireless Power Transfer},''
  \emph{IEEE Transactions on Signal Processing}, vol.~64, no.~23, pp.
  6313--6328, 2016.

\bibitem{zuhraITW2}
S.~u. Zuhra, S.~M. Perlaza, H.~V. Poor, and M.~Skoglund, ``{Information-Energy
  Trade-offs with EH Non-linearities in the Finite Block-Length Regime with
  Finite Constellations },'' in \emph{Proceedings of the {IEEE} Information
  Theory Workshop (ITW)}, Mumbai, India, Nov. 2022.

\bibitem{9241856}
N.~Shanin, L.~Cottatellucci, and R.~Schober, ``{Markov Decision Process Based
  Design of SWIPT Systems: Non-Linear EH Circuits, Memory, and Impedance
  Mismatch},'' \emph{IEEE Transactions on Communications}, vol.~69, no.~2, pp.
  1259--1274, 2021.

\bibitem{9447959}
S.~Abeywickrama, R.~Zhang, and C.~Yuen, ``{Refined Nonlinear Rectenna Modeling
  and Optimal Waveform Design for Multi-User Multi-Antenna Wireless Power
  Transfer},'' \emph{IEEE Journal of Selected Topics in Signal Processing},
  vol.~15, no.~5, pp. 1198--1210, 2021.

\bibitem{9411899}
S.~Shen and B.~Clerckx, ``{Joint Waveform and Beamforming Optimization for MIMO
  Wireless Power Transfer},'' \emph{IEEE Transactions on Communications},
  vol.~69, no.~8, pp. 5441--5455, 2021.

\bibitem{9184149}
O.~L.~A. López, F.~A. Monteiro, H.~Alves, R.~Zhang, and M.~Latva-Aho, ``{A
  Low-Complexity Beamforming Design for Multiuser Wireless Energy Transfer},''
  \emph{IEEE Wireless Communications Letters}, vol.~10, no.~1, pp. 58--62,
  2021.

\bibitem{7867826}
Y.~Zeng, B.~Clerckx, and R.~Zhang, ``{Communications and Signals Design for
  Wireless Power Transmission},'' \emph{IEEE Transactions on Communications},
  vol.~65, no.~5, pp. 2264--2290, 2017.

\bibitem{9447237}
P.~Mukherjee, C.~Psomas, and I.~Krikidis, ``{Differential Chaos Shift
  Keying-Based Wireless Power Transfer with Nonlinearities},'' \emph{IEEE
  Journal of Selected Topics in Signal Processing}, vol.~15, no.~5, pp.
  1185--1197, 2021.

\bibitem{9153166}
J.~Kim, B.~Clerckx, and P.~D. Mitcheson, ``{Signal and System Design for
  Wireless Power Transfer: Prototype, Experiment and Validation},'' \emph{IEEE
  Transactions on Wireless Communications}, vol.~19, no.~11, pp. 7453--7469,
  2020.

\bibitem{9377479}
A.~Khalili, S.~Zargari, Q.~Wu, D.~W.~K. Ng, and R.~Zhang, ``{Multi-Objective
  Resource Allocation for IRS-Aided SWIPT},'' \emph{IEEE Wireless
  Communications Letters}, vol.~10, no.~6, pp. 1324--1328, 2021.

\bibitem{9593249}
G.~M. Kraidy, C.~Psomas, and I.~Krikidis, ``{Fundamentals of Circular QAM for
  Wireless Information and Power Transfer},'' in \emph{Proceedings of the
  {IEEE} International Workshop on Signal Processing Advances in Wireless
  Communications (SPAWC)}, 2021, pp. 616--620.

\bibitem{survey}
T.~D.~P. Perera, D.~N.~K. Jayakody, S.~K. Sharma, S.~Chatzinotas, and J.~Li,
  ``{Simultaneous Wireless Information and Power Transfer ({SWIPT}): Recent
  Advances and Future Challenges},'' \emph{{IEEE} Communications Surveys \&
  Tutorials}, vol.~20, no.~1, pp. 264--302, Dec. 2018.

\bibitem{surveyA}
J.~Huang, C.-C. Xing, and C.~Wang, ``{Simultaneous Wireless Information and
  Power Transfer: Technologies, Applications, and Research Challenges},''
  \emph{{IEEE} Communications Magazine}, vol.~55, no.~11, pp. 26--32, Nov.
  2017.

\bibitem{8476597}
B.~Clerckx, R.~Zhang, R.~Schober, D.~W.~K. Ng, D.~I. Kim, and H.~V. Poor,
  ``{Fundamentals of Wireless Information and Power Transfer: From RF Energy
  Harvester Models to Signal and System Designs},'' \emph{IEEE Journal on
  Selected Areas in Communications}, vol.~37, no.~1, pp. 4--33, 2019.

\bibitem{ClerckxFoundations}
B.~Clerckx, J.~Kim, K.~W. Choi, and D.~I. Kim, ``{Foundations of Wireless
  Information and Power Transfer: Theory, Prototypes, and Experiments},''
  \emph{Proceedings of the IEEE}, vol. 110, no.~1, pp. 8--30, Jan. 2022.

\bibitem{6373669}
L.~Liu, R.~Zhang, and K.-C. Chua, ``{Wireless Information Transfer with
  Opportunistic Energy Harvesting},'' \emph{IEEE Transactions on Wireless
  Communications}, vol.~12, no.~1, pp. 288--300, 2013.

\bibitem{9149424}
N.~Shanin, L.~Cottatellucci, and R.~Schober, ``{Rate-Power Region of SWIPT
  Systems Employing Nonlinear Energy Harvester Circuits with Memory},'' in
  \emph{Proceedings of the {IEEE} International Conference on Communications
  (ICC)}, Dublin, Ireland, Jun. 2020, pp. 1--7.

\bibitem{7998252}
E.~Goudeli, C.~Psomas, and I.~Krikidis, ``{Sequential Decoding for Simultaneous
  Wireless Information and Power Transfer},'' in \emph{Proceedings of the
  International Conference on Telecommunications (ICT)}, 2017, pp. 1--5.

\bibitem{9741251}
G.~Lin, Y.~Zhou, W.~Jiang, X.~He, X.~Zhou, G.~He, and P.~Yang, ``{LF-SWIPT:
  Outage Analysis for SWIPT Relaying Networks Using Lossy Forwarding With QoS
  Guaranteed},'' \emph{IEEE Internet of Things Journal}, pp. 1--1, 2022.

\bibitem{9734045}
H.~T. Thien, P.-V. Tuan, and I.~Koo, ``{A Secure-Transmission Maximization
  Scheme for SWIPT Systems Assisted by an Intelligent Reflecting Surface and
  Deep Learning},'' \emph{IEEE Access}, vol.~10, pp. 31\,851--31\,867, 2022.

\bibitem{liu2022joint}
J.~Liu, C.-H.~R. Lin, Y.-C. Hu, and P.~K. Donta, ``{Joint Beamforming, Power
  Allocation, and Splitting Control for SWIPT-Enabled IoT Networks with Deep
  Reinforcement Learning and Game Theory},'' \emph{Sensors}, vol.~22, no.~6, p.
  2328, 2022.

\bibitem{9502719}
B.~Clerckx, K.~Huang, L.~R. Varshney, S.~Ulukus, and M.-S. Alouini, ``{Wireless
  Power Transfer for Future Networks: Signal Processing, Machine Learning,
  Computing, and Sensing},'' \emph{IEEE Journal of Selected Topics in Signal
  Processing}, vol.~15, no.~5, pp. 1060--1094, 2021.

\bibitem{6489506}
R.~Zhang and C.~K. Ho, ``{MIMO Broadcasting for Simultaneous Wireless
  Information and Power Transfer},'' \emph{IEEE Transactions on Wireless
  Communications}, vol.~12, no.~5, pp. 1989--2001, 2013.

\bibitem{7063588}
C.~Song, C.~Ling, J.~Park, and B.~Clerckx, ``{MIMO Broadcasting for
  Simultaneous Wireless Information and Power Transfer: Weighted MMSE
  Approaches},'' in \emph{Proceedings of the {IEEE} Globecom Workshops (GC
  Wkshps)}, 2014, pp. 1151--1156.

\bibitem{9169700}
M.~Varasteh, J.~Hoydis, and B.~Clerckx, ``{Learning to Communicate and
  Energize: Modulation, Coding, and Multiple Access Designs for Wireless
  Information-Power Transmission},'' \emph{IEEE Transactions on
  Communications}, vol.~68, no.~11, pp. 6822--6839, 2020.

\bibitem{tseV}
D.~Tse and P.~Viswanath, \emph{{Fundamentals of Wireless Communication}}.\hskip
  1em plus 0.5em minus 0.4em\relax Cambridge University Press, 2005.

\bibitem{lapidoth}
A.~Lapidoth, \emph{{A Foundation in Digital Communication}}, 2nd~ed.\hskip 1em
  plus 0.5em minus 0.4em\relax Cambridge University Press, 2017.

\bibitem{polyanskiy2010channel}
Y.~Polyanskiy, H.~V. Poor, and S.~Verd{\'u}, ``{Channel Coding Rate in the
  Finite Blocklength Regime},'' \emph{IEEE Transactions on Information Theory},
  vol.~56, no.~5, pp. 2307--2359, 2010.

\bibitem{proakis}
J.~G. Proakis and M.~Salehi, \emph{{Digital Communications}}, 5th~ed.\hskip 1em
  plus 0.5em minus 0.4em\relax McGraw-Hill Higher Education, 2008.

\bibitem{robbins1955remark}
H.~Robbins, ``{A Remark on {Stirling's} Formula},'' \emph{The American
  Mathematical Monthly}, vol.~62, no.~1, pp. 26--29, 1955.

\end{thebibliography}

 \clearpage
 \appendix
 \subsection{Proofs for Theorem~\ref{TheoremImpossibility}}
\label{AppendixA}
\begin{theorem} \label{LemmaImpossibleDEP}
Given a constant composition $(n,M,\mathcal{X},\epsilon,B,\delta)$-code $\mathscr{C}$ for the random transformation in~\eqref{EqYXdistribution} of the form in~\eqref{Eqnm_code}, for all $\ell \inCountK{L}$, let the complex $\bar{x}^{(\ell)} \in \mathcal{X}$ in~\eqref{EqCIsymbols} be such that
\begin{IEEEeqnarray}{rCl} \label{EqNeighbor}
\bar{x}^{(\ell)} \in \arg\max_{x \in \mathcal{X} \setminus \lbrace x^{(\ell)} \rbrace} \left|x^{(\ell)} - x \right|.
\end{IEEEeqnarray}
Then, $\epsilon$ in~\eqref{EqGammaUpperbound} must satisfy the following:
\begin{IEEEeqnarray}{l} \label{EqboundDEP}
\epsilon \geq (M-1) \mathrm{Q} \left( \sqrt{\frac{\sum_{\ell=1}^L n P_{\mathscr{C}}(x^{(\ell)}) \left| x^{(\ell)} - \bar{x}^{(\ell)} \right|^2}{2 \sigma^2}} \right), \IEEEeqnarraynumspace
\end{IEEEeqnarray}
where, $P_{\mathscr{C}}$ is the type defined in~\eqref{eq:p_bar}; $L$ is the number of symbols in~\eqref{EqCIsymbols}, and the real $\sigma^2$ is the noise variance in~\eqref{Eq4a}.
\end{theorem}
\begin{proof}
To determine the lower bound on the DEP, the decoding regions $\mathcal{D}_1$,  $\mathcal{D}_2$, $\ldots$, $\mathcal{D}_M$  are defined using the MAP decision rule~\cite[Chapter $21$]{lapidoth} as in~\eqref{EqMAPregion}.
Using the definition of $f_{\boldsymbol{Y}|\boldsymbol{X}}$ in~\eqref{EqDensities} followed simple mathematical manipulations, the decision rule in~\eqref{EqMAPregion} simplifies to the minimum distance decoder~\cite[Chapter 4]{proakis} which decides that message $i$ was transmitted if
\begin{IEEEeqnarray}{rCl} \label{Eq58b}
i \in \arg \min_{i' \inCountK{M}} \left| \boldsymbol{y} - \boldsymbol{u}(i') \right|.
\end{IEEEeqnarray}

The decoding error probability $\gamma_i \left( \mathscr{C} \right)$ in~\eqref{EqDEPi} for the code $\mathscr{C}$, given that the message index $i \inCountK{M}$ is transmitted is
\begin{IEEEeqnarray}{l} 
\gamma_i \left( \mathscr{C} \right) = \mathrm{Pr} \left( \boldsymbol{Y} \notin \mathcal{D}_i \big| \boldsymbol{u}(i) \right) \label{Eq662} \\
= \sum_{\substack{j =1 \\ j\neq i}}^M \mathrm{Pr} \left( \left| \boldsymbol{Y} - \boldsymbol{u}(j) \right|^2 < \left| \boldsymbol{Y} - \boldsymbol{u}(i) \right|^2 \big| \boldsymbol{u}(i) \right)  \label{Eq663} \\
= \sum_{\substack{j =1 \\ j\neq i}}^M \mathrm{Pr} \left( \left| \boldsymbol{u}(i) + \boldsymbol{N} - \boldsymbol{u}(j) \right|^2 < \left| \boldsymbol{u}(i) + \boldsymbol{N} - \boldsymbol{u}(i) \right|^2 \right) \\
= \sum_{\substack{j =1 \\ j\neq i}}^M \mathrm{Pr} \left( \left| \boldsymbol{u}(i) + \boldsymbol{N} - \boldsymbol{u}(j) \right|^2 < \left|\boldsymbol{N} \right|^2 \right) \\
= \sum_{\substack{j =1 \\ j\neq i}}^M \mathrm{Pr} \left( \sum_{m=1}^n \left| u_m(i) + N_m - u_m(j) \right|^2 < \sum_{m=1}^n \left| N_m \right|^2 \right) \\
= \sum_{\substack{j =1 \\ j\neq i}}^M \mathrm{Pr} \Big( \sum_{m=1}^n (\Re(u_m(i) + N_m - u_m(j)))^2 + \nonumber \\
(\Im(u_m(i) + N_m - u_m(j)))^2 < \sum_{m=1}^n \Re(N_m)^2 + \Im(N_m)^2 \Big) \\
= \sum_{\substack{j =1 \\ j\neq i}}^M \mathrm{Pr} \Big( \sum_{m=1}^n (\Re(u_m(i)) - \Re(u_m(j)))^2 + (\Im(u_m(i)) - \Im(u_m(j)))^2 \nonumber \\
+ 2 \Re(N_m)(\Re(u_m(i)) - \Re(u_m(j))) + \nonumber \\
2 \Im(N_m)(\Im(u_m(i)) - \Im(u_m(j))) < 0 \Big) \\
= \sum_{\substack{j =1 \\ j\neq i}}^M \mathrm{Pr} \Big( \sum_{m=1}^n \Re(N_m)(\Re(u_m(i)) - \Re(u_m(j))) + \nonumber \\
\Im(N_m)(\Im(u_m(i)) - \Im(u_m(j))) < - \frac{1}{2} \sum_{m=1}^n (\Re(u_m(i)) - \Re(u_m(j)))^2 \nonumber \\
+ (\Im(u_m(i)) - \Im(u_m(j)))^2 \Big), \label{Eq68b}
\end{IEEEeqnarray}
where, the equality in~\eqref{Eq663} follows from~\eqref{Eq58b}, and $\boldsymbol{N} = (N_{1}, N_{2}, \ldots, N_{n})^{\sf{T}}$ is the AWGN noise vector in~\eqref{EqChannelModel} such that, for all $m \inCountK{n}$, the random variable $N_{m}$ is a complex circularly symmetric Gaussian random variable whose real and imaginary parts have zero means and variances $\frac{1}{2}\sigma^2$. 
Therefore, the random variable 
\begin{IEEEeqnarray}{rCl} \label{Eq69}
\overline{N}_{i,j} &=& \sum_{m=1}^n \Re(N_m)(\Re(u_m(i)) - \Re(u_m(j))) \nonumber \\
&+& \Im(N_m)(\Im(u_m(i)) - \Im(u_m(j)))
\end{IEEEeqnarray}
in~\eqref{Eq68b} is a linear combination of $2n$ Gaussian random variables, each with mean zero and variance $\frac{1}{2}\sigma^2$. Thus, $\overline{N}_{i,j}$ is also a zero mean Gaussian random variable with variance given by  
\begin{IEEEeqnarray}{rCl}  
\sigma_{\overline{N}_{i,j}}^2 &=&  \frac{\sigma^2}{2} \sum_{m=1}^n (\Re(u_m(i)) - \Re(u_m(j)))^2 + \nonumber \\
&&(\Im(u_m(i)) - \Im(u_m(j)))^2. \label{Eq700}
\end{IEEEeqnarray}
From~\eqref{Eq68b},\eqref{Eq69} and~\eqref{Eq700}, it follows that,
\begin{IEEEeqnarray}{rCl} \label{Eq72}
\gamma_i \left( \mathscr{C} \right) &=& \sum_{\substack{j =1 \\ j\neq i}}^M \mathrm{Pr} \left( \overline{N}_{i,j} < - \frac{\sigma_{\overline{N}_{i,j}}^2}{\sigma^2} \right).
\end{IEEEeqnarray}
Since $\overline{N}_{i,j} \sim \mathcal{N}(0,\sigma_{\overline{N}_{i,j}}^2)$, from~\eqref{Eq72} it follows that
\begin{IEEEeqnarray}{l} 
\gamma_i \left( \mathscr{C} \right) =  \sum_{\substack{j =1 \\ j\neq i}}^M \mathrm{Q} \left( \frac{\sigma_{\overline{N}_{i,j}}}{\sigma^2} \right) \\
 = \sum_{\substack{j =1 \\ j\neq i}}^M \mathrm{Q} \left( \sqrt{\frac{1}{2 \sigma^2} \sum_{m=1}^n (\Re(u_m(i)) - \Re(u_m(j)))^2 + (\Im(u_m(i)) - \Im(u_m(j)))^2} \right) \\
= \sum_{\substack{j =1 \\ j\neq i}}^M \mathrm{Q} \left( \sqrt{\frac{\sum_{m=1}^n \left| u_m(i) - u_m(j) \right|^2}{2 \sigma^2}} \right), \label{Eq70}
\end{IEEEeqnarray}
where, the $\mathrm{Q}$ function is in~\eqref{DefQfunc}.
For all $i \inCountK{M}$ and all $m \inCountK{n}$, $u_m(i) \in \mathcal{X} = \lbrace x^{(1)}, x^{(2)}, \ldots, x^{(L)}\rbrace$ in~\eqref{EqCIsymbols}. For $u_m(i) = x^{(\ell)}$, from~\eqref{EqNeighbor} it follows that
\begin{IEEEeqnarray}{rCl} \label{Eq76a}
\left| u_m(i) - u_m(j) \right| \leq \left| x^{(\ell)} - \bar{x}^{(\ell)} \right|
\end{IEEEeqnarray}
Since $\mathscr{C}$ is a constant composition code, from~\eqref{eq:u_measure} and Definition~\ref{DefHC}, it follows that, for all $i \inCountK{M}$, the symbol $x^{(\ell)} \in \mathcal{X}$ appears in the codeword $\boldsymbol{u}(i)$, $n P_{\mathscr{C}}(x^{(\ell)})$ number of times.
Therefore, from~\eqref{Eq70} and~\eqref{Eq76a}, it follows that,
\begin{IEEEeqnarray}{l} 
\gamma_i \left( \mathscr{C} \right) \geq \sum_{\substack{j =1 \\ j\neq i}}^M \mathrm{Q} \left( \sqrt{\frac{\sum_{\ell=1}^L n P_{\mathscr{C}}(x^{(\ell)}) \left| x^{(\ell)} - \bar{x}^{(\ell)} \right|^2}{2 \sigma^2}} \right) \nonumber \\
= (M-1) \mathrm{Q} \left( \sqrt{\frac{\sum_{\ell=1}^L n P_{\mathscr{C}}(x^{(\ell)}) \left| x^{(\ell)} - \bar{x}^{(\ell)} \right|^2}{2 \sigma^2}} \right).
\end{IEEEeqnarray}
The average DEP $\gamma(\mathscr{C})$ in~\eqref{eq:gamma} is given by
\begin{IEEEeqnarray}{l}
 \gamma \left(\mathscr{C} \right)  \geq \frac{1}{M} \sum_{i=1}^M (M-1) \mathrm{Q} \left( \sqrt{\frac{\sum_{\ell=1}^L n P_{\mathscr{C}}(x^{(\ell)}) \left| x^{(\ell)} - \bar{x}^{(\ell)} \right|^2}{2 \sigma^2}} \right) \nonumber \\
   = (M-1) \mathrm{Q} \left( \sqrt{\frac{\sum_{\ell=1}^L n P_{\mathscr{C}}(x^{(\ell)}) \left| x^{(\ell)} - \bar{x}^{(\ell)} \right|^2}{2 \sigma^2}} \right). \label{Eq74}
        \end{IEEEeqnarray}
    
    Finally, from~\eqref{EqGammaUpperbound} and~\eqref{Eq74}, it follows that
\begin{IEEEeqnarray}{l}
\label{Eq699}
\epsilon \geq (M-1) \mathrm{Q} \left( \sqrt{\frac{\sum_{\ell=1}^L n P_{\mathscr{C}}(x^{(\ell)}) \left| x^{(\ell)} - \bar{x}^{(\ell)} \right|^2}{2 \sigma^2}} \right),
\end{IEEEeqnarray}
which completes the proof.
\end{proof}

The following lemma proves the bound on the information rate $R(\mathscr{C})$ in~\eqref{EqRbound}.
\begin{lemma} \label{lemma:R_upperbound}
Given a constant composition $(n,M,\mathcal{X},\epsilon,B,\delta)$-code $\mathscr{C}$ for the random transformation in~\eqref{EqChannelModel} of the form in~\eqref{Eqnm_code}, the  information transmission rate $R(\mathscr{C})$ in~\eqref{EqR} is such that
\begin{equation}
\label{EqRbound}
R(\mathscr{C}) \leq \frac{1}{n} \log_2 \left( \frac{n!}{\prod_{\ell=1}^L (nP_{\mathscr{C}}(x^{(\ell)}))!}\right) ,
\end{equation}
where, $P_{\mathscr{C}}$ is the type defined in~\eqref{eq:p_bar}.
\end{lemma}
\begin{proof}
Given a code type $P_{\mathscr{C}}$ that satisfies~\eqref{EqHomogeneousCodes}, the number of codewords that can be constructed is given by 
\begin{IEEEeqnarray}{l}
M = \binom{n}{nP_{\mathscr{C}}(x^{(1)})} \binom{n-nP_{\mathscr{C}}(x^{(1)})}{nP_{\mathscr{C}}(x^{(2)})}  \ldots \binom{n-\sum_{\ell=1}^{L-1} nP_{\mathscr{C}}(x^{(\ell)})}{nP_{\mathscr{C}}(x^{(L)})} \nonumber \\
= \frac{n!}{\prod_{\ell=1}^L (nP_{\mathscr{C}}(x^{(\ell)}))!}. \label{Rfactorials}
\end{IEEEeqnarray}
Therefore, the information rate $R(\mathscr{C})$ in~\eqref{EqR} satisfies
\begin{IEEEeqnarray}{rCl}
R(\mathscr{C}) &\leq& \frac{1}{n} \log_2 \left( \frac{n!}{\prod_{\ell=1}^L (nP_{\mathscr{C}}(x^{(\ell)}))!}\right),
\end{IEEEeqnarray}
which completes the proof.
\end{proof}

The next result provides proof of the bound on $\delta$ in~\eqref{EqdImp}.

\begin{theorem} \label{LemmaBhomogeneous}
Given a constant composition $(n,M,\mathcal{X},\epsilon,B,\delta)$-code $\mathscr{C}$ for the random transformation in~\eqref{EqYXdistribution} of the form in~\eqref{Eqnm_code}, the following holds for $\delta$ in~\eqref{eq:delta}:
\begin{IEEEeqnarray}{rCl} \label{EqboundEOP}
\delta \geq \mathds{1}_{\lbrace e_\mathscr{C} < B \rbrace},
\end{IEEEeqnarray}
where $e_\mathscr{C} \in [0,\infty)$ is in~\eqref{EqeiHomogeneous}.
\end{theorem}
\begin{proof}
From~\eqref{Eq25}, the average EOP for the $(n,M,\mathcal{X},\epsilon,B,\delta)$-code $\mathscr{C}$ is given by
\begin{IEEEeqnarray}{rCl}  \label{Eq40}
 \theta(\mathscr{C},B) = \frac{1}{M} \sum_{i=1}^M \mathds{1}_{\lbrace e_i < B \rbrace}.
 \end{IEEEeqnarray}
 Since $\mathscr{C}$ is a constant composition code, from~\eqref{Eq40} and~\eqref{EqeiHomogeneous} it follows that
 \begin{IEEEeqnarray}{rCl} 
 \theta(\mathscr{C},B) &=& \frac{1}{M} \sum_{i=1}^M \mathds{1}_{\lbrace e_\mathscr{C} < B \rbrace} \\
 &=& \mathds{1}_{\lbrace e_\mathscr{C} < B \rbrace}. \label{Eq42}
 \end{IEEEeqnarray}
 From~\eqref{eq:delta} and~\eqref{Eq42}, it follows that
\begin{IEEEeqnarray}{rCl} \label{Eq477}
\delta \geq \mathds{1}_{\lbrace e_\mathscr{C} < B \rbrace}.
\end{IEEEeqnarray} 
 This completes the proof.
\end{proof}

\subsection{
Proof of Theorem~\ref{CorRUpperBoundRelax}}
\label{AppendixC}
\begin{proof} 
From~\eqref{EqRbound}, for the information transmission rate $R(\mathscr{C})$ of code $\mathscr{C}$, it holds that
\begin{IEEEeqnarray}{rCl}
\label{EqRrelax}
R(\mathscr{C}) &\leq& \frac{1}{n} \log(n!) - \frac{1}{n} \sum_{\ell=1}^L \log\left(\left(nP_{\mathscr{C}}(x^{(\ell)})\right)!\right).
\end{IEEEeqnarray}
Using the Stirling's approximation~\cite{robbins1955remark} on the factorial terms yields
\begin{IEEEeqnarray}{l}
  (nP_{\mathscr{C}}(x^{(\ell)}))!  \geqslant \sqrt{2\pi} \left(nP_{\mathscr{C}}(x^{(\ell)}) \right)^{nP_{\mathscr{C}}(x^{(\ell)}) +\frac{1}{2}} \nonumber \\
\exp\left(-nP_{\mathscr{C}}(x^{(\ell)}) + \frac{1}{12nP_{\mathscr{C}}(x^{(\ell)}) +1} \right), \mbox{ and } \label{EqStirlingnP} \\
  \label{EqStirlingn}
  n !  \leqslant \sqrt{2\pi} n^{n +\frac{1}{2}} \exp\left(-n + \frac{1}{12 n} \right). 
\end{IEEEeqnarray}

From~\eqref{EqStirlingnP} and  ~\eqref{EqStirlingn}, it follows that,
\begin{IEEEeqnarray}{l}
 \log\left( (nP_{\mathscr{C}}(x^{(\ell)}))! \right) 
\geq  \log \left( \sqrt{2\pi}\right) + \left(nP_{\mathscr{C}}(x^{(\ell)}) +\frac{1}{2}\right) \nonumber \\
\log (nP_{\mathscr{C}}(x^{(\ell)})) - nP_{\mathscr{C}}(x^{(\ell)}) + \frac{1}{12nP_{\mathscr{C}}(x^{(\ell)}) +1} \\
\nonumber
=  \log \left( \sqrt{2\pi}\right) +  nP_{\mathscr{C}}(x^{(\ell)}) \log (P_{\mathscr{C}}(x^{(\ell)})) +\frac{1}{2} \log (P_{\mathscr{C}}(x^{(\ell)})) \\
+  \left(nP_{\mathscr{C}}(x^{(\ell)}) +\frac{1}{2}\right) \log (n)  -nP_{\mathscr{C}}(x^{(\ell)}) + \frac{1}{12nP_{\mathscr{C}}(x^{(\ell)}) +1}, \\
\nonumber \mbox{ and },\\
\nonumber
 \log\left( n! \right) \leq
\log \left( \sqrt{2\pi}\right) + \left(n +\frac{1}{2}\right) \log (n) -n + \frac{1}{12n}  \\
\label{EqRnfact}
= n  \log (n) - n + \frac{1}{12n} + \frac{1}{2}\log \left( 2\pi n\right).
 \end{IEEEeqnarray}
The sum in~\eqref{EqRrelax} satisfies, 
\begin{IEEEeqnarray}{l}
\nonumber
\sum_{\ell=1}^L \log\left(\left(nP_{\mathscr{C}}(x^{(\ell)})\right)!\right)
\geq L \log \left( \sqrt{2\pi}\right)  - nH\left( P_{\mathscr{C}} \right) \nonumber \\
+\frac{1}{2}\sum_{\ell=1}^L \log (P_{\mathscr{C}}(x^{(\ell)})) + n \log (n) +\frac{L}{2}\log(n) - n \nonumber \\
\label{EqRrelaxA}
 +\sum_{\ell=1}^L \frac{1}{12nP_{\mathscr{C}}(x^{(\ell)}) +1}. 
\end{IEEEeqnarray}
Using~\eqref{EqRnfact} and~\eqref{EqRrelaxA} in~\eqref{EqRrelax} yields,
\begin{IEEEeqnarray}{l}
\nonumber
R(\mathscr{C}) \leq \log (n) - 1 + \frac{1}{12n^2} + \frac{1}{2n}\log \left( 2\pi n\right) \\
- \frac{L}{n} \log \left( \sqrt{2\pi}\right)  + H\left( P_{\mathscr{C}} \right)  -\frac{1}{2n}\sum_{\ell=1}^L \log (P_{\mathscr{C}}(x^{(\ell)})) - \log (n) \nonumber \\
- \frac{L}{2n}\log(n) + 1 - \frac{1}{n}\sum_{\ell=1}^L \frac{1}{12nP_{\mathscr{C}}(x^{(\ell)}) +1} \\
\leq H\left( P_{\mathscr{C}} \right)  + \frac{1}{n^2} \left( \frac{1}{12} - \sum_{\ell=1}^L \frac{1}{12 P_{\mathscr{C}}(x^{(\ell)}) +1}\right)  + \nonumber \\
\frac{1}{2n} \left( \log \left( 2\pi n\right)   - \sum_{\ell=1}^L \log (2\pi n P_{\mathscr{C}}(x^{(\ell)}))  \right) \\
= H\left( P_{\mathscr{C}} \right)  + \frac{1}{n^2} \left( \frac{1}{12} - \sum_{\ell=1}^L \frac{1}{12 P_{\mathscr{C}}(x^{(\ell)}) +1}\right) +  \label{Eq90} \\
\frac{1}{n} \left(\log \left( \sqrt{2\pi}\right)   - \sum_{\ell=1}^L \log\sqrt{2\pi P_{\mathscr{C}}(x^{(\ell)})} \right) - \frac{\log n}{n} \left( \frac{L-1}{2} \right), \nonumber 
\end{IEEEeqnarray}
which completes the proof.
\end{proof}

\subsection{
Proof of Lemma~\ref{LemmaBnew}}
\label{AppendixE}
\begin{proof}
Given that channel input $i \inCountK{M}$ is transmitted, the harvested energy is given by the constant $e_i$ in~\eqref{Eqei}.
Furthermore, given the set of unique energy levels $\left\lbrace \bar{e}_1, \bar{e}_2, \ldots, \bar{e}_{M'} \right\rbrace$ in~\eqref{EqUniqueLevels}, there are a finite number of values that the EOP $\delta$ can take.
More precisely,
\begin{IEEEeqnarray}{rCl} 
\delta \in \left\lbrace \frac{\sum_{k=1}^{j}  \sum_{i=1}^M \mathds{1}_{\left\lbrace e_i = \bar{e}_k \right\rbrace }}{M}, j \inCountK{M'}  \right\rbrace. 
 \end{IEEEeqnarray}

This implies that, for all $j \inCountK{M'}$, for $\delta \leqslant \frac{\displaystyle\sum_{k=1}^{j}  \displaystyle\sum_{i=1}^M \mathds{1}_{\left\lbrace e_i = \bar{e}_k \right\rbrace }}{M}$, at most $\lfloor M \delta \rfloor = \displaystyle\sum_{k=1}^{j}  \displaystyle\sum_{i=1}^M \mathds{1}_{\left\lbrace e_i = \bar{e}_k \right\rbrace }$ codewords can have energy less than $B$. This is possible if and only if
\begin{IEEEeqnarray}{rCl} 
\label{EqBrd}
B \leq \bar{e}_{j} .
\end{IEEEeqnarray}

Define $j^{+} \in \ints$ as 
\begin{IEEEeqnarray}{rCl}
\label{Eqjplus}
j^{+} & \triangleq & \min \left\lbrace j \in \lbrace 1, \ldots M' \rbrace: \delta \leqslant \frac{\displaystyle\sum_{k=1}^{j}  \displaystyle\sum_{i=1}^M \mathds{1}_{\left\lbrace e_i = \bar{e}_k \right\rbrace }}{M} \right\rbrace,
\end{IEEEeqnarray}
where, the positive integer $M'$ and the reals $\bar{e}_1$, $\bar{e}_2$, $\ldots$, $\bar{e}_{M'}$ are in~\eqref{EqUniqueLevels}.

Then, from~\eqref{EqBrd} and~\eqref{Eqjplus}, it follows that
\begin{IEEEeqnarray}{rCl} 
B \leq \bar{e}_{j^{+}} 
\end{IEEEeqnarray}
This completes the proof.
\end{proof}

 \subsection{Proofs for Theorem~\ref{TheoremAchievableRegion}}  \label{SecAchievableBounds} 
The following lemma provides a lower bound on the DEP of codes from the family ${\sf C} \left(C,\boldsymbol{A},\boldsymbol{L},\boldsymbol{\alpha},\boldsymbol{p},\boldsymbol{r} \right)$ in~\eqref{EqFamily}.

\begin{lemma} \label{LemmaRadius}
Consider an $(n,M,\mathcal{X},\epsilon,B,\delta)$-code $\mathscr{C}$ from the family ${\sf C} \left(C,\boldsymbol{A},\boldsymbol{L},\boldsymbol{\alpha},\boldsymbol{p},\boldsymbol{r} \right)$ in~\eqref{EqFamily}. The parameters $r_1, r_2, \ldots, r_C$ in~\eqref{EqRadiusVector} satisfy the following:
\begin{equation}
\epsilon \geq 1 - \frac{1}{M}\sum_{i = 1}^M\prod_{c=1}^{C}  \left( 1-\exp \left(-\frac{r_c^2}{\sigma^2} \right) \right)^{n \sum_{\ell=1}^{L_c} P_{\boldsymbol{u}(i)}(x_c^{(\ell)})},
\end{equation}
where, the type  $P_{\boldsymbol{u}(i)}$ is defined in~\eqref{eq:u_measure}, the real $\sigma^2$ is defined in~\eqref{EqDensities}, and $x_c^{(\ell)} \in \mathcal{U}(A_c,L_c,\alpha_c)$, with $\mathcal{U}(A_c,L_c,\alpha_c)$ in \eqref{EqLayerCircle}.
\end{lemma}
\begin{proof}
From~\eqref{EqDEPi} and~\eqref{eq:gamma}, the average DEP of the $\left(n,M,\mathcal{X},\epsilon,B,\delta \right)$-code $\mathscr{C}$ is given by
\begin{IEEEeqnarray}{l}
    \gamma(\mathscr{C}) = 1 - \frac{1}{M} \sum_{i=1}^M \int_{\mathcal{D}_i} f_{\boldsymbol{Y}|\boldsymbol{X}}(\boldsymbol{y}|\boldsymbol{u}(i)) \mathrm{d}\boldsymbol{y} \\
    = 1 - \frac{1}{M} \sum_{i=1}^M \int_{\mathcal{D}_{i,1} \times \mathcal{D}_{i,2} \times \ldots \times \mathcal{D}_{i,n}} \prod_{m=1}^n f_{Y|X} \left( y|u_m(i) \right) \mathrm{d}y \label{Eq138b} \\
= 1 - \frac{1}{M} \sum_{i=1}^M \prod_{m=1}^n \int_{\mathcal{D}_{i,m}} f_{Y|X} \left( y|u_m(i) \right) \mathrm{d}y \label{Eq139} \\ 
= 1 - \frac{1}{M} \sum_{i=1}^M \prod_{c=1}^C \prod_{\ell=1}^{L_c} \left(  \int_{\mathcal{G}_c^{(\ell)}} f_{Y|X}(y|x_c^{(\ell)}) \mathrm{d}y  \right)^{n P_{\boldsymbol{u}(i)}(x_c^{(\ell)})}, \label{Eq140}
\end{IEEEeqnarray}
where, the equality in~\eqref{Eq138b} follows due to~\eqref{EqYXdistribution} and~\eqref{EqDecodingCodeword}, and~\eqref{Eq139} follows due to Fubini's theorem.
Using~\eqref{EqDensities} in~\eqref{Eq140} yields,
\begin{IEEEeqnarray}{l} \label{EqAvgGamma}
    \gamma \left(\mathscr{C} \right) = 1 - \frac{1}{M} \sum_{i=1}^M \prod_{c=1}^C \prod_{\ell=1}^{L_c} \Bigg(  \int_{\mathcal{G}_c^{(\ell)}} \frac{1}{\pi \sigma^2} \\
    \exp \left(- \frac{(\Re(y)-\Re(x_c^{(\ell)}))^2 +(\Im(y) - \Im(x_c^{(\ell)}))^2}{\sigma^2} \right) \mathrm{d}y  \Bigg)^{n P_{\boldsymbol{u}(i)}(x_c^{(\ell)})}. \nonumber
\end{IEEEeqnarray}
Evaluating the integral term in~\eqref{EqAvgGamma} for all $\ell \inCountK{L_c}$ yields,
\begin{IEEEeqnarray}{l}
\int_{\mathcal{G}_c^{(\ell)}} \frac{1}{\pi \sigma^2} \exp \left(- \frac{(\Re(y)-\Re(x_c^{(\ell)}))^2 +(\Im(y) - \Im(x_c^{(\ell)}))^2}{\sigma^2} \right) \mathrm{d}y \nonumber \\
= \int_{\Im(x_c^{(\ell)})-r_c}^{\Im(x_c^{(\ell)})+r_c} \int_{\Re(x_c^{(\ell)})-\sqrt{r_c^2-(v-\Im(x_c^{(\ell)}))^2}}^{\Re(x_c^{(\ell)})+\sqrt{r_c^2-(v-\Im(x_c^{(\ell)}))^2}} \nonumber \\
\frac{1}{\pi \sigma^2} \exp \left(-\frac{(u-\Re(x_c^{(\ell)}))^2+(v-\Im(x_c^{(\ell)}))^2}{\sigma^2} \right) \mathrm{d}u \mathrm{d}v,\\
= \int_{-r_c}^{r_c} \int_{-\sqrt{r_c^2-v^2}}^{\sqrt{r_c^2-v^2}} \frac{1}{\pi \sigma^2} \exp \left(-\frac{u^2+v^2}{\sigma^2} \right) \mathrm{d}u \mathrm{d}v, \label{Eq82}\\
= \int_{0}^{\pi} \int_{0}^{\frac{r_c}{\sigma}} \frac{1}{2\pi} \exp \left(-\zeta^2 \right) \zeta \mathrm{d}\zeta \mathrm{d}\eta, \label{Eq83} \\
= \left( 1-\exp \left(-\frac{r_c^2}{\sigma^2}\right) \right). \label{Eq94}
\end{IEEEeqnarray}
The equality in~\eqref{Eq83} is obtained from the change of variables $u = \sigma \zeta \cos \eta, v = \sigma \zeta \sin \eta$. Plugging~\eqref{Eq94} in~\eqref{EqAvgGamma} yields,
\begin{IEEEeqnarray}{rCl}
\gamma \left(\mathscr{C} \right) &=& 1 - \frac{1}{M} \sum_{i=1}^M \prod_{c=1}^C \prod_{\ell=1}^{L_c} \left( 1-\exp \left(-\frac{r_c^2}{\sigma^2}\right) \right)^{n P_{\boldsymbol{u}(i)}(x_c^{(\ell)})} ,\\
 \label{Eq95}    &=& 1 - \frac{1}{M}\sum_{i = 1}^M \prod_{c=1}^C \left( 1-\exp \left(-\frac{r_c^2}{\sigma^2} \right) \right)^{n \sum_{\ell=1}^{L_c} P_{\boldsymbol{u}(i)}(x_c^{(\ell)})}.
\end{IEEEeqnarray}
From~\eqref{EqGammaUpperbound}, for $\mathscr{C}$ to be an $\left( n,M,\mathcal{X},\epsilon \right)$-code, the following must hold:
\begin{equation} \label{EqGammaSumBound}
    \gamma \left(\mathscr{C} \right) \leq \epsilon.
\end{equation}
This implies that,
\begin{IEEEeqnarray}{rCl}
\epsilon \geq 1 - \frac{1}{M}\sum_{i = 1}^M \prod_{c=1}^C \left( 1-\exp \left(-\frac{r_c^2}{\sigma^2} \right) \right)^{n \sum_{\ell=1}^{L_c} P_{\boldsymbol{u}(i)}(x_c^{(\ell)})},
    \end{IEEEeqnarray}
which completes the proof.
\end{proof}
The bound on $\epsilon$ given by Lemma~\ref{LemmaRadius} provides the minimum value of the DEP that can be achieved by the constructed family of codes.
This implies that any DEP requirement greater than this value can be satisfied by the constructed family of codes.

From Definition~\ref{DefHC} and~\eqref{EqTypeCircle}, it follows that, for a constant composition code $\mathscr{C}$ of the form in~\eqref{EqnmCodeCircle} from the family ${\sf C} \left(C,\boldsymbol{A},\boldsymbol{L},\boldsymbol{\alpha},\boldsymbol{p},\boldsymbol{r} \right)$ in~\eqref{EqFamily}, it holds that,
\begin{IEEEeqnarray}{rCl}
    \label{Eq100}
   P_{\boldsymbol{u}(i)}(x_c^{(\ell)}) = P_{\mathscr{C}}(x_c^{(\ell)}) = \frac{p_c}{L_c}.
\end{IEEEeqnarray} 

The following result for constant composition codes follows from~\eqref{Eq100} and Lemma~\ref{LemmaRadius}.
%
\begin{corollary}
\label{CorRadiusHomogeneous}
Consider a constant composition $(n,M,\mathcal{X},\epsilon,B,\delta)$-code $\mathscr{C}$ from the family ${\sf C} \left(C,\boldsymbol{A},\boldsymbol{L},\boldsymbol{\alpha},\boldsymbol{p},\boldsymbol{r} \right)$ in~\eqref{EqFamily}. The parameters $r_1, r_2, \ldots, r_C$ in~\eqref{EqRadiusVector} satisfy the following:
\begin{equation}
 \epsilon \geq 1- \prod_{c=1}^{C}  \left( 1-\exp \left(-\frac{r_c^2}{\sigma^2} \right) \right)^{n p_c},
\end{equation}
where, the real $\sigma^2$ is defined in~\eqref{EqDensities}, and $x_c^{(\ell)} \in \mathcal{U}(A_c,L_c,\alpha_c)$, with $\mathcal{U}(A_c,L_c,\alpha_c)$ in \eqref{EqLayerCircle}.
\end{corollary}

In the following lemma, we provide an upper bound on the achievable information transmission rate
$R\left(\mathscr{C} \right)$ for a code $\mathscr{C}$ from the family ${\sf C} \left(C,\boldsymbol{A},\boldsymbol{L},\boldsymbol{\alpha},\boldsymbol{p},\boldsymbol{r} \right)$ in~\eqref{EqFamily}. 
Proving Lemma~\ref{LemmaAchievableRnew} requires some additional results that provide bounds on the radius of decoding regions $r_c$ in~\eqref{EqRadiusVector} as well as the number of symbols in each layer of the constructed codes $L_c$ in~\eqref{EqLcVector}.
In the interest of better readability, these results are presented as Lemmas~\ref{CorollaryRadius} and~\ref{LemmaLc} in Appendix~\ref{AppendixLemmas} and we directly state the bound on the information transmission rate $R\left(\mathscr{C} \right)$ here.
\begin{lemma} \label{LemmaAchievableRnew}
Given an $(n,M,\mathcal{X},\epsilon,B,\delta)$-code $\mathscr{C}$ from the family ${\sf C} \left(C,\boldsymbol{A},\boldsymbol{L},\boldsymbol{\alpha},\boldsymbol{p},\boldsymbol{r} \right)$ in~\eqref{EqFamily}, the information transmission rate $R\left(\mathscr{C} \right)$ satisfies the following:
\begin{IEEEeqnarray}{rCl}
    R\left(\mathscr{C} \right) \leq \log_2 \sum_{c=1}^C \left\lfloor \frac{\pi}{2\arcsin{\frac{r_c}{2A_c}}} \right\rfloor,
\end{IEEEeqnarray}
where, for all $c \inCountK{C}$, the radius $r_c$ is in~\eqref{EqRadiusVector}, and the amplitude $A_c$ is in~\eqref{EqLayerCircle}.
\end{lemma}
\begin{proof}
The largest number of codewords of length $n$ that can be formed with $L$ different channel input symbols is $L^n$. Hence, from~\eqref{EqR}, it follows that
\begin{IEEEeqnarray}{rCl}
    R\left(\mathscr{C} \right) &\leq& \frac{\log_2 M}{n} \\
    &\leq& \frac{\log_2 L^n}{n} \\
    &=& \log_2 L \label{EqRStep1} \\
    &\leq& \log_2 \sum_{c=1}^C \left\lfloor \frac{\pi}{2\arcsin{\frac{r_c}{2A_c}}} \right\rfloor, \label{EqRStep2}
\end{IEEEeqnarray}
where, the inequality in~\eqref{EqRStep2} follows from Lemma~\ref{LemmaLc} and~\eqref{EqLSum}. This completes the proof.
\end{proof}
The following lemma provides a bound on the information transmission rate for a constant composition $(n,M,\mathcal{X},\epsilon,B,\delta)$-code.
\begin{lemma} \label{LemmaAchievableR}
For a constant composition $(n,M,\mathcal{X},\epsilon,B,\delta)$-code $\mathscr{C}$ from the family ${\sf C} \left(C,\boldsymbol{A},\boldsymbol{L},\boldsymbol{\alpha},\boldsymbol{p},\boldsymbol{r} \right)$ in~\eqref{EqFamily}, the information transmission rate $R\left(\mathscr{C} \right)$ is given by:
\begin{IEEEeqnarray}{rCl} \label{Eq174b}
    R\left(\mathscr{C} \right) \leq \frac{1}{n} \log_2 \left( \frac{n!}{\prod_{c=1}^C \left( \left( n\frac{ p_c}{L_c} \right)! \right)^{L_c}} \right).
\end{IEEEeqnarray}
\end{lemma}
\begin{proof}
%
For a constant composition $(n,M,\mathcal{X},\epsilon,B,\delta)$-code $\mathscr{C}$ from the family ${\sf C} \left(C,\boldsymbol{A},\boldsymbol{L},\boldsymbol{\alpha},\boldsymbol{p},\boldsymbol{r} \right)$ for which the type $P_{\mathscr{C}}$ satisfies~\eqref{EqHomogeneousCodes}, the number of codewords that can be constructed  is given by 
{\allowdisplaybreaks
\begin{IEEEeqnarray}{l}
M = \binom{n}{nP_{\mathscr{C}}(x_1^{(1)})} \binom{n-nP_{\mathscr{C}}(x_1^{(1)})}{nP_{\mathscr{C}}(x_1^{(2)})}  \ldots \binom{n-\sum_{\ell=1}^{L_1-1} nP_{\mathscr{C}}(x_1^{(\ell)})}{nP_{\mathscr{C}}(x_1^{(L_1)})} \nonumber \\
\times \binom{n-\sum_{\ell=1}^{L_1} nP_{\mathscr{C}}(x_1^{(\ell)})}{nP_{\mathscr{C}}(x_2^{(1)})} \binom{n-\sum_{\ell=1}^{L_1} nP_{\mathscr{C}}(x_1^{(\ell)}) - nP_{\mathscr{C}}(x_2^{(1)})}{nP_{\mathscr{C}}(x_2^{(2)})} \nonumber \\
\ldots \binom{n-\sum_{\ell=1}^{L_1} nP_{\mathscr{C}}(x_1^{(\ell)})-\sum_{\ell=1}^{L_2-1} nP_{\mathscr{C}}(x_2^{(\ell)})}{nP_{\mathscr{C}}(x_2^{(L_2)})}  \times \ldots \nonumber \\
\times \binom{n-\sum_{c=1}^{C-1} \sum_{\ell=1}^{L_c} nP_{\mathscr{C}}(x_c^{(\ell)})}{nP_{\mathscr{C}}(x_C^{(1)})} 
 \binom{n-\sum_{c=1}^{C-1} \sum_{\ell=1}^{L_c} nP_{\mathscr{C}}(x_c^{(\ell)}) - nP_{\mathscr{C}}(x_C^{(1)})}{nP_{\mathscr{C}}(x_C^{(2)})} \nonumber \\
 \ldots \binom{n-\sum_{c=1}^{C-1} \sum_{\ell=1}^{L_c} nP_{\mathscr{C}}(x_c^{(\ell)})-\sum_{\ell=1}^{L_C-1} nP_{\mathscr{C}}(x_C^{(\ell)})}{nP_{\mathscr{C}}(x_C^{(L_C)})}  \\
 = \binom{n}{n\frac{p_1}{L_1}} \binom{n-n\frac{p_1}{L_1}}{n\frac{p_1}{L_1}}  \ldots \binom{n-\sum_{\ell=1}^{L_1-1} n\frac{p_1}{L_1}}{n\frac{p_1}{L_1}}\nonumber \\
 \times 
\binom{n-np_1}{n\frac{p_2}{L_2}}
\binom{n-np_1 - n\frac{p_2}{L_2}}{n\frac{p_2}{L_2}} \ldots \binom{n-np_1-\sum_{\ell=1}^{L_2-1} n\frac{p_2}{L_2}}{n\frac{p_2}{L_2}}  \times \nonumber \\
\ldots \times \binom{n-\sum_{c=1}^{C-1} np_c}{n\frac{p_C}{L_C}}
 \binom{n-\sum_{c=1}^{C-1} np_c - n\frac{p_C}{L_C}}{n\frac{p_C}{L_C}} \ldots \nonumber \\
 \binom{n-\sum_{c=1}^{C-1} np_c-\sum_{\ell=1}^{L_C-1} n\frac{p_C}{L_C}}{n\frac{p_C}{L_C}}  \\
 = \frac{n!}{\prod_{c=1}^C \left( \left( n\frac{ p_c}{L_c} \right)! \right)^{L_c}}
\end{IEEEeqnarray}
}
Therefore, the information transmission rate $R\left(\mathscr{C} \right)$ is given by
\begin{IEEEeqnarray}{rCl} \label{Eq159b}
    R\left(\mathscr{C} \right) \leq \frac{\log_2 M}{n} = \frac{1}{n} \log_2 \left( \frac{n!}{\prod_{c=1}^C \left( \left( n\frac{ p_c}{L_c} \right)! \right)^{L_c}} \right). \qquad
\end{IEEEeqnarray}
This completes the proof.
\end{proof}

The following lemma provides the achievable bound on the EOP $\delta$ in~\eqref{EqTheoremrc} for codes from the family ${\sf C} \left(C,\boldsymbol{A},\boldsymbol{L},\boldsymbol{\alpha},\boldsymbol{p},\boldsymbol{r} \right)$.
\begin{lemma} \label{LemmaAchievableEOP}
For an $(n,M,\mathcal{X},\epsilon,B,\delta)$-code $\mathscr{C}$ from the family ${\sf C} \left(C,\boldsymbol{A},\boldsymbol{L},\boldsymbol{\alpha},\boldsymbol{p},\boldsymbol{r} \right)$ in~\eqref{EqFamily}, the following holds:
\begin{IEEEeqnarray}{rCl} \label{Eq213}
 \delta \geq  \frac{1}{M} \sum_{i=1}^M \mathds{1}_{\left\lbrace \sum_{c=1}^C \sum_{\ell=1}^{L_c} P_{\boldsymbol{u}(i)} \left( x_c^{(\ell)} \right) \left( k_1 A_c^2 + k_2  A_c^4 \right) < \frac{B}{n} \right\rbrace},
 \end{IEEEeqnarray}
where, $k_1$ and $k_2$ are positive real constants defined in~\eqref{Eqei} and $P_{\boldsymbol{u}(i)}$ is the type in~\eqref{eq:u_measure}.
\end{lemma}
%
\begin{proof}
From~\eqref{EqLayerCircle}, for all $c \inCountK{C}$ and all $\ell \inCountK{L_c}$, the symbols $x_c^{\left( \ell \right)} \in \mathcal{U}(A_c,L_c,\alpha_c)$ in~\eqref{EqLayerCircle} are given by
\begin{IEEEeqnarray}{rCl} 
x_c^{\left( \ell \right)} = A_c \exp\left(\mathrm{i} \left(\frac{2\pi}{L_c} \ell+\alpha_c\right)\right).
\end{IEEEeqnarray}
This implies that
\begin{IEEEeqnarray}{rCl} \label{Eq218}
\left| x_c^{\left( \ell \right)} \right| = A_c.
\end{IEEEeqnarray}
Using~\eqref{Eq25} and~\eqref{Eqei}, the EOP for the $(n,M,\mathcal{X},\epsilon,B,\delta)$-code $\mathscr{C}$ is given by:
\begin{IEEEeqnarray}{l} 
 \theta(\mathscr{C},B) = \frac{1}{M} \sum_{i=1}^M \mathds{1}_{\left\lbrace ( k_1 \sum_{c=1}^C \sum_{\ell=1}^{L_c} n P_{\boldsymbol{u}(i)} \left( x_c^{(\ell)} \right) \left| x_c^{(\ell)} \right|^2 + k_2 \sum_{c=1}^C \sum_{\ell=1}^{L_c} n P_{\boldsymbol{u}(i)} \left( x_c^{(\ell)} \right) \left| x_c^{(\ell)} \right|^4 ) < B \right\rbrace} \\
= \frac{1}{M} \sum_{i=1}^M \mathds{1}_{\left\lbrace ( k_1 \sum_{c=1}^C \sum_{\ell=1}^{L_c} n P_{\boldsymbol{u}(i)} \left( x_c^{(\ell)} \right) A_c^2 + k_2 \sum_{c=1}^C \sum_{\ell=1}^{L_c} n P_{\boldsymbol{u}(i)} \left( x_c^{(\ell)} \right) A_c^4) < B \right\rbrace}. \label{Eq209}
 \end{IEEEeqnarray}
From~\eqref{eq:delta} and~\eqref{Eq209}, the code $\mathscr{C}$ is an $\left( n,M,\mathcal{X},\epsilon,B,\delta \right)$-code if the following holds:
\begin{IEEEeqnarray}{rCl} 
  \frac{1}{M} \sum_{i=1}^M \mathds{1}_{\left\lbrace \sum_{c=1}^C \sum_{\ell=1}^{L_c} P_{\boldsymbol{u}(i)} \left( x_c^{(\ell)} \right) \left( k_1 A_c^2 + k_2  A_c^4 \right) < \frac{B}{n} \right\rbrace} \leq \delta.
\end{IEEEeqnarray}
This completes the proof.
\end{proof}
The achievable bound on the EOP $\delta$ in~\eqref{eq:delta} for a constant composition code $\mathscr{C}$ from the family ${\sf C} \left(C,\boldsymbol{A},\boldsymbol{L},\boldsymbol{\alpha},\boldsymbol{p},\boldsymbol{r} \right)$ in~\eqref{EqFamily} is given by the following lemma:
\begin{lemma} \label{LemmaAmplitudeC}
For a constant composition $(n,M,\mathcal{X},\epsilon,B,\delta)$-code $\mathscr{C}$ from the family ${\sf C} \left(C,\boldsymbol{A},\boldsymbol{L},\boldsymbol{\alpha},\boldsymbol{p},\boldsymbol{r} \right)$ in~\eqref{EqFamily}, for all $c \inCountK{C}$, the parameters $p_c$ in~\eqref{Eqpc} satisfy the following:
\begin{IEEEeqnarray}{rCl}
 \delta = \mathds{1}_{\left\lbrace \left( k_1 \sum_{c=1}^C n p_c A_c^2 + k_2 \sum_{c=1}^C n p_c A_c^4 \right) < B \right\rbrace},
 \end{IEEEeqnarray}
where, $k_1$ and $k_2$ are positive real constants defined in~\eqref{Eqei}.
\end{lemma}
\begin{proof}
From~\eqref{Eq25} and~\eqref{EqeiHomogeneous}, the EOP for the constant composition code $\mathscr{C}$ is given by:
\begin{IEEEeqnarray}{l} 
 \theta(\mathscr{C},B) = \frac{1}{M} \sum_{i=1}^M \mathds{1}_{\left\lbrace \substack{ ( k_1 \sum_{c=1}^C \sum_{\ell=1}^{L_c} n P_{\mathscr{C}} \left( x_c^{(\ell)} \right) \left| x_c^{(\ell)} \right|^2 + \\ k_2 \sum_{c=1}^C \sum_{\ell=1}^{L_c} n P_{\mathscr{C}} \left( x_c^{(\ell)} \right) \left| x_c^{(\ell)} \right|^4 ) < B} \right\rbrace} \\
= \mathds{1}_{\left\lbrace \substack{ ( k_1 \sum_{c=1}^C \sum_{\ell=1}^{L_c} n P_{\mathscr{C}} \left( x_c^{(\ell)} \right) A_c^2 + \\ k_2 \sum_{c=1}^C \sum_{\ell=1}^{L_c} n P_{\mathscr{C}} \left( x_c^{(\ell)} \right) A_c^4) < B} \right\rbrace} \label{Eq191}\\
= \mathds{1}_{\left\lbrace \left( k_1 \sum_{c=1}^C \sum_{\ell=1}^{L_c} n \frac{p_c}{L_c} A_c^2 + k_2 \sum_{c=1}^C \sum_{\ell=1}^{L_c} n \frac{p_c}{L_c} A_c^4 \right) < B \right\rbrace} \label{Eq192} \\
= \mathds{1}_{\left\lbrace \left( k_1 \sum_{c=1}^C n p_c A_c^2 + k_2 \sum_{c=1}^C n p_c A_c^4 \right) < B \right\rbrace}. \label{Eq266}
 \end{IEEEeqnarray}
 The equality in~\eqref{Eq191} follows from~\eqref{Eq218} and~\eqref{Eq192} follows from~\eqref{EqTypeCircle}.
From~\eqref{eq:delta} and~\eqref{Eq266}, it follows that the code $\mathscr{C}$ is an $\left( n,M,\mathcal{X},\epsilon,B,\delta \right)$-code if the following holds:
\begin{IEEEeqnarray}{rCl} 
 \delta = \mathds{1}_{\left\lbrace \left( k_1 \sum_{c=1}^C n p_c A_c^2 + k_2 \sum_{c=1}^C n p_c A_c^4 \right) < B \right\rbrace}.
\end{IEEEeqnarray}
This completes the proof.
\end{proof}

The following lemma provides an upper bound on the energy transmission rate $B$ for $(n,M,\mathcal{X},\epsilon,B,\delta)$-codes from the family ${\sf C} \left(C,\boldsymbol{A},\boldsymbol{L},\boldsymbol{\alpha},\boldsymbol{p},\boldsymbol{r} \right)$.

\begin{lemma} \label{LemmaBAchievable}
For an $(n,M,\mathcal{X},\epsilon,B,\delta)$-code $\mathscr{C}$ from the family ${\sf C} \left(C,\boldsymbol{A},\boldsymbol{L},\boldsymbol{\alpha},\boldsymbol{p},\boldsymbol{r} \right)$ in~\eqref{EqFamily}, let $j^{+} \in \ints$ be 
\begin{IEEEeqnarray}{rcl}
j^{+} & \triangleq & \min \left\lbrace j \in \lbrace 1, \ldots M' \rbrace: \delta \leqslant \frac{\displaystyle\sum_{k=1}^{j}  \displaystyle\sum_{i=1}^M \mathds{1}_{\left\lbrace e_i = \bar{e}_k \right\rbrace }}{M} \right\rbrace,
\end{IEEEeqnarray}
where, the positive integer $M'$ and the reals $\bar{e}_1$, $\bar{e}_2$, $\ldots$, $\bar{e}_{M'}$ are in~\eqref{EqUniqueLevels}.
Then, the following holds for the energy transmission rate $B$:
\begin{IEEEeqnarray}{rCl} \label{EqBAchievable}
B \leq \bar{e}_{j^{+}} 
\end{IEEEeqnarray}
\end{lemma}
%
%
\begin{proof}
The proof follows on the same lines as that for Lemma~\ref{LemmaBnew} where, for all $i \inCountK{M}$, the energy $e_i$ in~\eqref{Eqei} is given by
\begin{IEEEeqnarray*}{l}
e_i = \sum_{c=1}^C \sum_{\ell=1}^{L_c} n P_{\boldsymbol{u}(i)} \left( x_c^{(\ell)} \right) A_c^2 + k_2 \sum_{c=1}^C \sum_{\ell=1}^{L_c} n P_{\boldsymbol{u}(i)} \left( x_c^{(\ell)} \right) A_c^4.
\end{IEEEeqnarray*}
\end{proof}

 \subsection{Miscellaneous Results} \label{AppendixLemmas}
For the special case where for all $c \inCountK{C}$, the radius of the decoding regions $r_c = r$ in~\eqref{EqDecodingCircleComplex}, the following lemma provides a lower bound on the value of $r$.
\begin{lemma} \label{CorollaryRadius}
Consider an $(n,M,\mathcal{X},P,\epsilon,B,\delta)$-code $\mathscr{C}$ from the family ${\sf C} \left(C,\boldsymbol{A},\boldsymbol{L},\boldsymbol{\alpha},\boldsymbol{p},\boldsymbol{r} \right)$ in~\eqref{EqFamily} such that, for all $c \inCountK{C}$, $r_c = r$ in~\eqref{EqRadiusVector}. Then, the parameter $r$ satisfies:
\begin{equation}\label{EqRminEqual}
    r \geq \sqrt{\sigma^2 \log \left( \frac{1}{1-(1-\epsilon)^{\frac{1}{n}}} \right) },
\end{equation}
where, the real $\sigma^2$ is defined in~\eqref{EqDensities}.
\end{lemma}
The lower bound on $r$ in Lemma~\ref{CorollaryRadius} characterizes the relationship between the DEP requirement $\epsilon$ and the spacing between channel input symbols $r$.
This is a key element in the code construction since it helps determine the precise structure of the constellation that can guarantee a given DEP.
Another crucial insight from Lemma~\ref{CorollaryRadius} is the relationship between the noise variance $\sigma^2$ and $r$.
It is known from classical information theoretic literature that the channel noise variance dictates how clearly two neighboring symbols can be distinguished at the decoder.
While this is a well known result, the precise relationship between these quantities cannot be determined due to the complexity of calculating the DEP.
The unique choice of decoding regions made in~\eqref{EqDecodingCircleComplex} allows us to arrive at~\eqref{EqRminEqual}, which clearly shows the trade-off between these two quantities and provides the required minimum distance between channel input symbols as a function of the noise variance $\sigma^2$.
\begin{proof}
If the parameters $r_c$ in~\eqref{EqRadiusVector} are such that, for all $c \inCountK{C}$, $r_c = r$, the average DEP in~\eqref{Eq95} is given by:
\begin{IEEEeqnarray}{rCl}
\gamma \left( \mathscr{C} \right) &=& 1 - \frac{1}{M}\sum_{i = 1}^M \prod_{c=1}^C \left( 1-e^{-\frac{r^2}{\sigma^2}} \right)^{n \sum_{\ell=1}^{L_c} P_{\boldsymbol{u}(i)}(x_c^{(\ell)})}, \\
&=& 1 - \frac{1}{M}\sum_{i = 1}^M \left( 1-e^{-\frac{r^2}{\sigma^2}} \right)^{n \sum_{c=1}^C \sum_{\ell=1}^{L_c} P_{\boldsymbol{u}(i)}(x_c^{(\ell)})}, \\
&=& 1 - \frac{1}{M}\sum_{i = 1}^M \left( 1-e^{-\frac{r^2}{\sigma^2}} \right)^{n}, \\
&=& 1 - \left( 1-e^{-\frac{r^2}{\sigma^2}} \right)^{n}. \label{Eq105}
\end{IEEEeqnarray}
From~\eqref{EqGammaUpperbound}, for $\mathscr{C}$ to be an $\left( n,M,\mathcal{X},P,\epsilon \right)$-code, the following must hold:
\begin{IEEEeqnarray}{rCl}
\epsilon \geq 1 - \left( 1-e^{-\frac{r^2}{\sigma^2}} \right)^{n}.
    \end{IEEEeqnarray}
This implies that
\begin{IEEEeqnarray}{rCl}
    r \geq \sqrt{\sigma^2 \log \left( \frac{1}{1-(1-\epsilon)^{\frac{1}{n}}} \right) }.
\end{IEEEeqnarray}
This completes the proof.
\end{proof}
\begin{figure}
    \centering
\tikzset{every picture/.style={line width=0.75pt}} 

\begin{tikzpicture}[x=0.75pt,y=0.75pt,yscale=-1,xscale=1]

\draw  [color={rgb, 255:red, 0; green, 0; blue, 0 }  ,draw opacity=1 ] (189.67,167) .. controls (189.67,101.28) and (242.94,48) .. (308.67,48) .. controls (374.39,48) and (427.67,101.28) .. (427.67,167) .. controls (427.67,232.72) and (374.39,286) .. (308.67,286) .. controls (242.94,286) and (189.67,232.72) .. (189.67,167) -- cycle ;
\draw    (310.67,324) -- (310.67,5) ;
\draw [shift={(310.67,2)}, rotate = 450] [fill={rgb, 255:red, 0; green, 0; blue, 0 }  ][line width=0.08]  [draw opacity=0] (8.93,-4.29) -- (0,0) -- (8.93,4.29) -- cycle    ;
\draw    (138.77,167.08) -- (484.67,169.48) ;
\draw [shift={(487.67,169.5)}, rotate = 180.4] [fill={rgb, 255:red, 0; green, 0; blue, 0 }  ][line width=0.08]  [draw opacity=0] (8.93,-4.29) -- (0,0) -- (8.93,4.29) -- cycle    ;
\draw  [fill={rgb, 255:red, 248; green, 231; blue, 28 }  ,fill opacity=0.27 ][dash pattern={on 4.5pt off 4.5pt}] (394.67,167) .. controls (394.67,148.77) and (409.44,134) .. (427.67,134) .. controls (445.89,134) and (460.67,148.77) .. (460.67,167) .. controls (460.67,185.23) and (445.89,200) .. (427.67,200) .. controls (409.44,200) and (394.67,185.23) .. (394.67,167) -- cycle ;
\draw [color={rgb, 255:red, 245; green, 166; blue, 35 }  ,draw opacity=1 ]   (309.8,168) -- (422.61,134) ;
\draw [color={rgb, 255:red, 245; green, 166; blue, 35 }  ,draw opacity=1 ]   (422.61,134) -- (427.67,167) ;
\draw [color={rgb, 255:red, 245; green, 166; blue, 35 }  ,draw opacity=1 ]   (309.8,169) -- (422.61,199.6) ;
\draw [color={rgb, 255:red, 245; green, 166; blue, 35 }  ,draw opacity=1 ]   (427.67,167) -- (422.61,199.6) ;
\draw  [dash pattern={on 0.84pt off 2.51pt}]  (221.63,86.08) -- (308.67,167) ;
\draw  [dash pattern={on 0.84pt off 2.51pt}]  (427.67,167) -- (447.13,193.58) ;
\draw  [dash pattern={on 4.5pt off 4.5pt}] (197.67,255) .. controls (197.67,236.77) and (212.44,222) .. (230.67,222) .. controls (248.89,222) and (263.67,236.77) .. (263.67,255) .. controls (263.67,273.23) and (248.89,288) .. (230.67,288) .. controls (212.44,288) and (197.67,273.23) .. (197.67,255) -- cycle ;
\draw  [dash pattern={on 4.5pt off 4.5pt}] (161.67,199) .. controls (161.67,180.77) and (176.44,166) .. (194.67,166) .. controls (212.89,166) and (227.67,180.77) .. (227.67,199) .. controls (227.67,217.23) and (212.89,232) .. (194.67,232) .. controls (176.44,232) and (161.67,217.23) .. (161.67,199) -- cycle ;
\draw  [dash pattern={on 4.5pt off 4.5pt}] (324.67,58) .. controls (324.67,39.77) and (339.44,25) .. (357.67,25) .. controls (375.89,25) and (390.67,39.77) .. (390.67,58) .. controls (390.67,76.23) and (375.89,91) .. (357.67,91) .. controls (339.44,91) and (324.67,76.23) .. (324.67,58) -- cycle ;
\draw  [dash pattern={on 4.5pt off 4.5pt}] (374.67,103) .. controls (374.67,84.77) and (389.44,70) .. (407.67,70) .. controls (425.89,70) and (440.67,84.77) .. (440.67,103) .. controls (440.67,121.23) and (425.89,136) .. (407.67,136) .. controls (389.44,136) and (374.67,121.23) .. (374.67,103) -- cycle ;
\draw  [dash pattern={on 4.5pt off 4.5pt}] (161.67,133) .. controls (161.67,114.77) and (176.44,100) .. (194.67,100) .. controls (212.89,100) and (227.67,114.77) .. (227.67,133) .. controls (227.67,151.23) and (212.89,166) .. (194.67,166) .. controls (176.44,166) and (161.67,151.23) .. (161.67,133) -- cycle ;
\draw  [dash pattern={on 4.5pt off 4.5pt}] (197.67,77) .. controls (197.67,58.77) and (212.44,44) .. (230.67,44) .. controls (248.89,44) and (263.67,58.77) .. (263.67,77) .. controls (263.67,95.23) and (248.89,110) .. (230.67,110) .. controls (212.44,110) and (197.67,95.23) .. (197.67,77) -- cycle ;
\draw  [dash pattern={on 4.5pt off 4.5pt}] (258.67,50) .. controls (258.67,31.77) and (273.44,17) .. (291.67,17) .. controls (309.89,17) and (324.67,31.77) .. (324.67,50) .. controls (324.67,68.23) and (309.89,83) .. (291.67,83) .. controls (273.44,83) and (258.67,68.23) .. (258.67,50) -- cycle ;
\draw  [dash pattern={on 4.5pt off 4.5pt}] (256.67,284) .. controls (256.67,265.77) and (271.44,251) .. (289.67,251) .. controls (307.89,251) and (322.67,265.77) .. (322.67,284) .. controls (322.67,302.23) and (307.89,317) .. (289.67,317) .. controls (271.44,317) and (256.67,302.23) .. (256.67,284) -- cycle ;
\draw  [dash pattern={on 4.5pt off 4.5pt}] (376.67,232) .. controls (376.67,213.77) and (391.44,199) .. (409.67,199) .. controls (427.89,199) and (442.67,213.77) .. (442.67,232) .. controls (442.67,250.23) and (427.89,265) .. (409.67,265) .. controls (391.44,265) and (376.67,250.23) .. (376.67,232) -- cycle ;
\draw  [dash pattern={on 4.5pt off 4.5pt}] (324.67,276) .. controls (324.67,257.77) and (339.44,243) .. (357.67,243) .. controls (375.89,243) and (390.67,257.77) .. (390.67,276) .. controls (390.67,294.23) and (375.89,309) .. (357.67,309) .. controls (339.44,309) and (324.67,294.23) .. (324.67,276) -- cycle ;

\draw (422.42,108.03) node [anchor=north west][inner sep=0.75pt]  [font=\large] [align=left] {{\fontfamily{ptm}\selectfont {\footnotesize \textit{S}}}};
\draw (298,168) node [anchor=north west][inner sep=0.75pt]  [font=\large] [align=left] {{\fontfamily{ptm}\selectfont {\footnotesize \textit{O}}}};
\draw (431.35,155) node [anchor=north west][inner sep=0.75pt]  [font=\large] [align=left] {{\fontfamily{ptm}\selectfont {\footnotesize \textit{T}}}};
\draw (418.82,204.23) node [anchor=north west][inner sep=0.75pt]  [font=\large] [align=left] {{\fontfamily{ptm}\selectfont {\footnotesize \textit{U}}}};
\draw (265.29,102.71) node [anchor=north west][inner sep=0.75pt]  [font=\normalsize] [align=left] {$\displaystyle A_c$};
\draw (440.79,169.21) node [anchor=north west][inner sep=0.75pt]  [font=\normalsize] [align=left] {$\displaystyle r_c$};
\draw (473,173) node [anchor=north west][inner sep=0.75pt]   [align=left] {$\displaystyle u$};
\draw (321,1) node [anchor=north west][inner sep=0.75pt]   [align=left] {$\displaystyle v$};
\draw (344,156) node [anchor=north west][inner sep=0.75pt]   [align=left] {$\displaystyle \beta $};

\draw (178,147.9) node [anchor=north west][inner sep=0.75pt]  [font=\large] [align=left] {{\fontfamily{ptm}\selectfont {\footnotesize \textit{V}}}};
\draw (185,165) node [anchor=north west][inner sep=0.75pt]  [font=\huge] [align=left] {{\fontfamily{ptm}\selectfont \textcolor[rgb]{0.82,0.01,0.11}{.}}};

\draw (418.31,131) node [anchor=north west][inner sep=0.75pt]  [font=\huge] [align=left] {{\fontfamily{ptm}\selectfont \textcolor[rgb]{0.82,0.01,0.11}{.}}};
\draw (423,167) node [anchor=north west][inner sep=0.75pt]  [font=\huge] [align=left] {{\fontfamily{ptm}\selectfont \textcolor[rgb]{0.82,0.01,0.11}{.}}};
\draw (418.31,198) node [anchor=north west][inner sep=0.75pt]  [font=\huge] [align=left] {{\fontfamily{ptm}\selectfont \textcolor[rgb]{0.82,0.01,0.11}{.}}};

\end{tikzpicture}

    \caption{Graphical representation of the symbols in layer $c$ defined in~\eqref{EqLayerCircle}}
    \label{FigConstellation}
\end{figure}
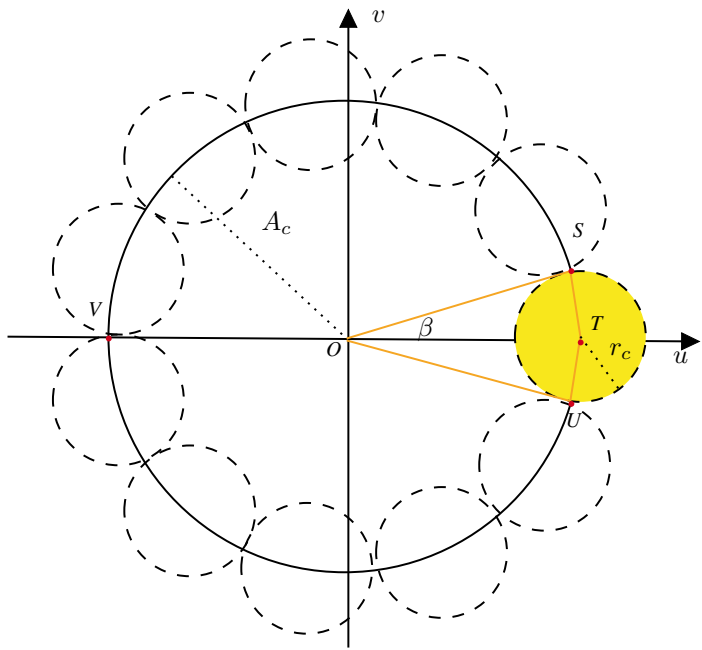

The information rate achievable by a code is a function of the number of channel input symbols $L$ in~\eqref{EqLSum} which in turn is a function of the number of symbols in each layer of the set of channel inputs $\mathcal{X}$ in~\eqref{EqConstellationCircle}. The following lemma provides an upper bound on the number of symbols in a layer $c \inCountK{C}$ denoted by $L_c$.
\begin{lemma} \label{LemmaLc}
Consider an $(n,M,\mathcal{X},P,\epsilon,B,\delta)$-code $\mathscr{C}$ from the family ${\sf C} \left(C,\boldsymbol{A},\boldsymbol{L},\boldsymbol{\alpha},\boldsymbol{p},\boldsymbol{r} \right)$ in~\eqref{EqFamily}. Then, for all $c \inCountK{C}$, the number of symbols in layer $c$ of $\mathcal{X}$ must satisfy the following:
\begin{equation} \label{Eq159n}
    L_c \leq \left\lfloor \frac{\pi}{2\arcsin{\frac{r_c}{2A_c}}} \right\rfloor,
\end{equation}
where, $r_c$ is the radius of the decoding regions in~\eqref{EqDecodingCircleComplex} and $A_c$ is the amplitude in~\eqref{EqLayerCircle}.
\end{lemma}
Lemma~\ref{LemmaLc} provides an upper bound on the number of symbols that can be packed into a given layer of the proposed code construction while respecting the DEP requirement $\epsilon$ in~\eqref{EqGammaUpperbound}.
Using~\eqref{Eq159n}, the maximum number of symbols $L_c$ that should be placed in layer $c$ of the code can be precisely determined.
For all $c \inCountK{C}$, the number of symbols in a layer $L_c$ is a function of the radius of decoding sets $r_c$ and the amplitude of the layer $A_c$.
Note that the radii $r_c$ are in turn functions of the noise variance $\sigma^2$ and the DEP requirement $\epsilon$ as shown in Lemma~\ref{CorollaryRadius}.
\begin{proof}
For an $(n,M,\mathcal{X},P,\epsilon,B,\delta)$-code $\mathscr{C}$ in the family ${\sf C} \left(C,\boldsymbol{A},\boldsymbol{L},\boldsymbol{\alpha},\boldsymbol{p},\boldsymbol{r} \right)$, for all $c \inCountK{C}$, the radius $r_c$ of the decoding regions in~\eqref{EqRadiusVector} and the amplitude $A_c$ in~\eqref{EqAcVector} determine the number of symbols $L_c$ that can be accommodated in the layer $c$. 
\par
Layer $c$ of the form in~\eqref{EqLayerCircle} is illustrated in Fig.~\ref{FigConstellation}. 
The symbols in layer $c$ are distributed uniformly along the circle of radius $A_c$ centered at the origin $O$.
The maximum number of symbols in layer $c$ is equal to the number of non-overlapping circles of radius $r_c$ corresponding to the decoding regions defined in~\eqref{EqDecodingCircleComplex} that can be placed along the circumference of the circle of radius $A_c$.
From Fig.~\ref{FigConstellation}, a circle of radius $r_c$ centered at a symbol in layer $c$ subtends angle $\angle \mbox{SOU} = \beta$ at $O$. Therefore, the maximum number of symbols $L_c$ that can be accommodated along the circle of radius $A_c$ is given by
    \begin{equation} \label{Eq109}
        L_c \leq \left\lfloor  \frac{2\pi}{\beta} \right\rfloor.
    \end{equation}
%

%
 To determine the value of the angle $\beta$ in Fig.~\ref{FigConstellation}, consider the circle of radius $r_c$ centered at $T$ and the larger circle of radius $A_c$ centered at the origin $O$. The circles intersect at points $S$ and $U$. The angle subtended by the major arc $\arc{\mbox{SU}}$ at $O$ is the reflex angle $2\pi - \beta$. Since the angle subtended by an arc of a circle at its centre is two times the angle that it subtends anywhere on the circumference, it holds that
\begin{equation}
    2\pi- \beta = 2\angle \mbox{STU},
\end{equation}
which implies that
\begin{equation} \label{EqHalfAngle}
    \angle \mbox{STU} = \frac{2\pi- \beta}{2}.
\end{equation}
The line segment $TO$ bisects angles $\angle \mbox{STU}$ and $\angle \mbox{SOU}$. Therefore, the following hold:
   \begin{IEEEeqnarray}{rCl}
        \angle \mbox{STO} &=& \frac{2\pi - \beta}{4}, \\
       \angle \mbox{SOT} &=& \frac{\beta}{2}.
   \end{IEEEeqnarray}
   From the triangle $\triangle \mbox{SOT}$, it holds that:
   \begin{IEEEeqnarray}{rCl}
  \frac{\sin \left(\angle \mbox{SOT} \right)}{ST} = \frac{\sin \left(\angle \mbox{STO} \right)}{SO}.
       \end{IEEEeqnarray}
       This implies that,
     \begin{IEEEeqnarray}{rCl}    
\frac{\sin \left(\frac{\beta}{2} \right) }{r_c} &=& \frac{\sin \left(\frac{2\pi-\beta}{4} \right)}{A_c}, \\
	&=& \frac{1}{A_c} \sin \left(\frac{\pi}{2} - \frac{\beta}{4} \right), \\
	&=&  \frac{1}{A_c} \cos \left( \frac{\beta}{4} \right). \label{Eq127a}
 \end{IEEEeqnarray}
From~\eqref{Eq127a}, it follows that,
       \begin{IEEEeqnarray}{rCl}    	
	  \frac{2}{r_c} \sin \left(\frac{\beta}{4} \right) \cos \left(\frac{\beta}{4} \right) = \frac{1}{A_c} \cos \left( \frac{\beta}{4} \right), 
	  \end{IEEEeqnarray}
	  which implies that,
	 \begin{IEEEeqnarray}{rCl}  
	 && \sin \left(\frac{\beta}{4} \right)  = \frac{r_c}{2A_c},  \quad \mbox{and}\\
       &&\beta = 4\arcsin{\frac{r_c}{2A_c}}. \label{Eq110}
    \end{IEEEeqnarray}

Substituting the value of $\beta$ from~\eqref{Eq110} in~\eqref{Eq109}, it follows that the number of symbols in layer $c$ of $\mathcal{X}$ is at most
    \begin{equation} \label{Eq126}
        L_c \leq \left\lfloor \frac{\pi}{2\arcsin{\frac{r_c}{2A_c}}} \right\rfloor.
    \end{equation}
This completes the proof.
\end{proof}

%
%
%
%
%
%
%
%

\end{document}